\documentclass{article}

\usepackage{arxiv}

\usepackage{times}
\usepackage{soul}
\usepackage{xurl}
\usepackage[hidelinks]{hyperref}
\usepackage[utf8]{inputenc}
\usepackage[small]{caption}
\usepackage{amsmath}
\usepackage{amsthm}
\usepackage{amssymb}
\usepackage{algorithm}
\usepackage{algorithmic}
\usepackage[switch]{lineno}
\usepackage{subcaption}

\usepackage[utf8]{inputenc} % allow utf-8 input
\usepackage[T1]{fontenc}    % use 8-bit T1 fonts
\usepackage{hyperref}       % hyperlinks
\usepackage{url}            % simple URL typesetting
\usepackage{booktabs}       % professional-quality tables
\usepackage{amsfonts}       % blackboard math symbols
\usepackage{nicefrac}       % compact symbols for 1/2, etc.
\usepackage{microtype}      % microtypography
\usepackage{lipsum}
\usepackage{graphicx}
\graphicspath{ {./images/} }

\newcommand{\mevgame}{\texttt{RST-Game}}
\newcommand{\ouralgo}{\texttt{RSYP}}
\newcommand{\shap}{\texttt{SHAP}}

\newtheorem{theorem}{Theorem}
\newtheorem{definition}{Definition}

\newtheorem{conjecture}{Conjecture}

\title{Shapley Value-based Approach for Redistributing Revenue of Matchmaking of Private Transactions in Blockchains}

\author{
 Rasheed \\
  Machine Learning Lab\\
  IIIT Hyderabad\\
  Hyderabad, India\\
  \texttt{mohammad.ahmed@research.iiit.ac.in}\\
   \And
 Parth Desai \\
  Machine Learning Lab\\
  IIIT Hyderabad\\
  Hyderabad, India\\
  \texttt{parth.desai@research.iiit.ac.in}\\
  \And
 Yash Chaurasia \\
  Machine Learning Lab\\
  IIIT Hyderabad\\
  Hyderabad, India\\
  \texttt{cayshvk@protonmail.com}\\
  \And
Sujt Gujar \\
  Machine Learning Lab\\
  IIIT Hyderabad\\
  Hyderabad, India\\
  \texttt{sujit.gujar@iiit.ac.in}\\
}

\begin{document}
\maketitle
\begin{abstract}
In the context of blockchain, MEV refers to the maximum value that can be extracted from block production through the inclusion, exclusion, or reordering of transactions. Searchers often participate in order flow auctions (OFAs) to obtain exclusive rights to private transactions, available through entities called \emph{matchmakers}, also known as order flow providers (OFPs). Most often, redistributing the revenue generated through such auctions among transaction creators is desirable. In this work, we formally introduce the \emph{matchmaking} problem in MEV, its desirable properties, and associated challenges. Using cooperative game theory, we formalize the notion of fair revenue redistribution in matchmaking and present its potential possibilities and impossibilities. Precisely, we define a characteristic form game, referred to as \mevgame, for the transaction creators. We propose to redistribute the revenue using the Shapley value of \mevgame. We show that the corresponding problem could be SUBEXP (i.e. $2^{o(n)}$, where $n$ is the number of transactions); therefore, approximating the Shapley value is necessary. Further, we propose a randomized algorithm for computing the Shapley value in \mevgame\ and empirically verify its efficacy. 
\end{abstract}

% keywords can be removed

\section{Introduction}
\label{sec:intro}
%%%%%%%%%%%%%%%%%%%%%%%%%%%%%%%%%%%%%%%%

% %%%%%%%%%%%%%%%%%%%%%%%
% \subsection{Motivation}
% \label{ssec:motivation}
% %%%%%%%%%%%%%%%%%%%%%%%
As defined in ~\cite{swan2015blockchain}, \emph{blockchain} is a distributed ledger of digital transactions maintained by a network of computers (or nodes) that reach consensus using one of a variety of mechanisms, such as \emph{Proof-of-Work} (PoW) ~\cite{nakamoto2008bitcoin} or \emph{Proof-of-Stake} (PoS)~\cite{buterin2013ethereum}. \emph{Maximal Extractable Value} (MEV) in the blockchain context refers to ``the maximum value that can be extracted from block production in excess of the standard block reward and gas fees by including, excluding, and changing the order of transactions in a block''~\cite{ethereumMaximalExtractable}. MEV extraction is an integral part of the Ethereum blockchain -- it enhances revenue, supports sustainable post-block rewards, and contributes to the overall health of DeFi activities~\cite{chainMaximalExtractable},~\cite{burian2024futuremev}, ~\cite{torres2024rollingshadowsanalyzingextraction}. However, MEV extraction possesses certain negative externalities such as worse execution price to transaction creators due to MEV attacks such as front-running and sandwiching~\cite{luganodesLuganodesWhat}, consensus-security risks, centralization due to shift in economic structure among blockchain participants, and legal and ethical concerns~\cite{flashbots},~\cite{ramosMev},~\cite{ji2024regulatory}. 

As MEV extraction becomes increasingly vital for sustaining the Ethereum blockchain's economic ecosystem, developments in the Ethereum's landscape~\cite{chaurasia2024mev}, such as its transition to Proof-of-Stake, have reformed how transactions are processed. This shift introduced the concept of \emph{proposer-builder separation} (PBS), where \emph{builders} are responsible for gathering transactions from the network, assembling them into blocks, and then bidding these blocks to the proposer, who holds the authority to publish them on the blockchain~\cite{gupta2023centralizingeffectsprivateorder}. Builders compete in an auction, known as the \emph{PBS auction}, for the opportunity to get their block included in the current slot~\cite{frontierBuilderDominance},~\cite{flashbotsSearchingPostMerge}. Proposers are heavily invested in the staking process and prioritize earning through staking rewards rather than generating profits through block building and MEV extraction. \emph{Searchers} are individual or institutional entities that actively seek to capitalize on profitable opportunities arising from various factors such as market inefficiencies during periods of high volatility~\cite{berg2022empirical}, the price of execution driven by supply-demand dynamics, and delays in transaction processing~\cite{thecryptocortexUnderstandingMarket}. They use techniques like front-running, back-running, sandwiching, etc~\cite{chi2024remeasuring},~\cite{wang2022impact} to extract MEV. Occasionally, some participants take on dual roles as both builders and searchers to maximize their profits~\cite{gupta2023centralizingeffectsprivateorder}.
 
% Note that it is possible for builders to also capture such profitable opportunities \cite{structuralAdvantages}.
The current paper focuses on searchers and \emph{transaction creators}. Typically, searchers find MEV extraction opportunities by analyzing unconfirmed transactions in the public mempool. Searchers strategically create a bundle of transactions to earn profit and bid (proportional to their potential profit) for block space with block builders. Unconfirmed transactions in a public mempool are accessible to all the blockchain participants, resulting in high competition among searchers
% and reducing the chances of winning the blockspace and PBS auction
~\cite{flashbotsOrderFlow}. Hence, searchers look to include private transactions to increase their chances of winning~\cite{gupta2023centralizingeffectsprivateorder}. Private transactions, also called private order flows, are transactions that are not available in the public mempool. 

Such MEV-rich private transactions are made available for sale by \emph{order flow providers} (OFPs), which typically include wallet service providers. OFPs collect intents from transaction creators, process them into transactions, and auction them off to searchers in what is known as \emph{order flow auctions} (OFAs)~\cite{structuralAdvantages}. While this transaction pipeline is profitable for searchers, there are two crucial points: (i) it might result in worse execution prices for transaction creators~\cite{luganodesLuganodesWhat}, and (ii) some OFPs charge transaction creators for providing their wallet services and may aim to enhance the experience for transaction creators by sharing the revenue generated from the OFAs with them.
This has led to the mechanism of \emph{Matchmaking}.

Matchmaking is a mechanism to redistribute the revenue generated by the order flow auctions to the transaction creators in exchange for the value that can be extracted from their transactions~\cite{flashbotsSearchingPostMerge}. A \emph{matchmaker} in blockchain is an OFP with access to transaction creators' private transactions. If the transactions are directly sent to the builders enforcing no MEV extraction, the potential MEV that could have been generated in the system from the transactions is lost. So, the matchmaker instead auctions off the transactions to searchers for efficient extraction of MEV. In return, the matchmaker compensates the transaction creators for the value their transactions create~\cite{match},~\cite{share}. Figure~\ref{fig::blocknative} shows that as of 2024, private transactions made up more than 30\% of all smart contract transactions on Ethereum (which was nearly 3 out of every 10 transactions). Though multiple authors have raised a need for such redistribution, as claimed in ~\cite{flashbotsMEVShareProgrammably},~\cite{flashbotsFRP30Quantifying}, designing a matchmaking mechanism is an open problem. This paper addresses how a matchmaker should redistribute the revenue generated through an OFA among transaction creators in a \emph{fair} manner. Some transactions add more value to the system or receive worse execution prices than others; thus, sharing the revenue equally is unfair. It should be shared proportional to how much value they add to the system. Hence, the Shapley values of an appropriately defined game become a natural way of redistributing revenue back to transaction creators.

In this paper, we define a cooperative game \mevgame\ over transaction creators based on the revenue generated in the system. We study \mevgame\ for two important types of valuations of the searchers-: (i) \emph{additive} and (ii) \emph{single-minded}. Additive value settings mean the total value a searcher derives from the transactions shared with it, is the addition of value for individual transactions. 

Searchers may be interested in a bundle(s) of transactions, which is the more practical situation in the MEV world. We abstract it through a case of single-minded searchers. Here, each searcher is interested in a particular bundle of transactions. A single-minded searcher prefers not to be allocated any transactions than to be allocated a strict subset of his interested bundle of transactions. It does not mind being allocated a superset of its interested bundle. In general, computing the Shapley value of a game is EXP~\cite{deng1994complexity}. However, the researchers have shown that for certain games, it could be computed efficiently, e.g.,~\cite{van2023efficiently},~\cite{michalak2013efficient}. We provide efficient computation even for \mevgame\ for additive value settings. But, for single-minded cases, it is difficult to compute Shapley value efficiently. We provide examples where computing the Shapley value of \mevgame\ is possibly subexponential in terms of the number of transactions and exponential in terms of the number of searchers. Hence, we conjecture that the Shapley value of \mevgame\ with single-minded searchers is SUBEXP \footnote{Of the two commonly used definitions of SUBEXP, we use the following: SUBEXP(n) = $2^{o(n)}$}. Since an analytical solution to the \mevgame\ is infeasible to compute for a large number of transactions or searchers, we resort to approximating the Shapley values. We design a polynomial time algorithm called Randomized ShapleY Procedure (\ouralgo) to compute the Shapley value of all transaction creators. $\ouralgo$ computes the marginal contribution of a transaction for randomly selected $\mathcal{O}(n^2)$ subsets of transactions than all possible subsets and approximate Shapley value as an appropriately weighted combination of these marginal contributions. We experimentally validate \ouralgo's computing efficiency and approximation. In summary, the following are our contributions. 

% The worth of each coalition of transaction creators is defined as the revenue generated in the case that only their transactions are present in the system. Given such a game, we propose redistributing the revenue among the transaction creators according to the Shapley value of their respective transactions. Computing Shapley value, in general, is at least SUBEXP~\cite{faigle1992shapley, deng1994complexity}. However, for certain applications, it can be computed in polynomial time by exploiting the underlying assumption, simplifying the Shapley computation~\cite{michalak2013efficient, hu2023computing}. We explore if Shapley values can be computed efficiently in the \mevgame. We propose to use the approximation method via random permutation sampling method and estimate the Shapley Value of players. We empirically show the Shapley estimate via a sample size of $\mathcal{O}(n^2)$, where $n$ is the number of transaction creators, closely resembles the actual Shapley value. We believe that worst cases are rare, and the proposed algorithm performs well statistically on average; thus, it can be used to deploy in real-time.
% We conjecture that computing the Shapley value in \mevgame\ in SUBEXP.

\noindent\textbf{Contributions.} (i) We define a cooperative game, \mevgame, over transaction creators based on the revenue generated, (ii) we prove that the Shapley value of transaction creators in the \mevgame\ is polynomial-time computable when the searcher valuations are additive, (iii) we motivate that computing Shapley value in the \mevgame\ when the searchers are single-minded bidders is possibly SUBEXP, (iv) we propose a randomized algorithm -- Randomized ShapleY Procedure (\ouralgo) that closely estimates the exact Shapley value of transaction creators, (v) we empirically show the efficacy of our algorithm by comparing its outputs with the brute-force approach.

We believe our results provide valuable insights for the practical deployment of matchmaking in the future. We leave it for future work to examine what guarantees one can give around fair redistribution among transaction creators.

\begin{figure}
    \centering
    \includegraphics[width=\linewidth]{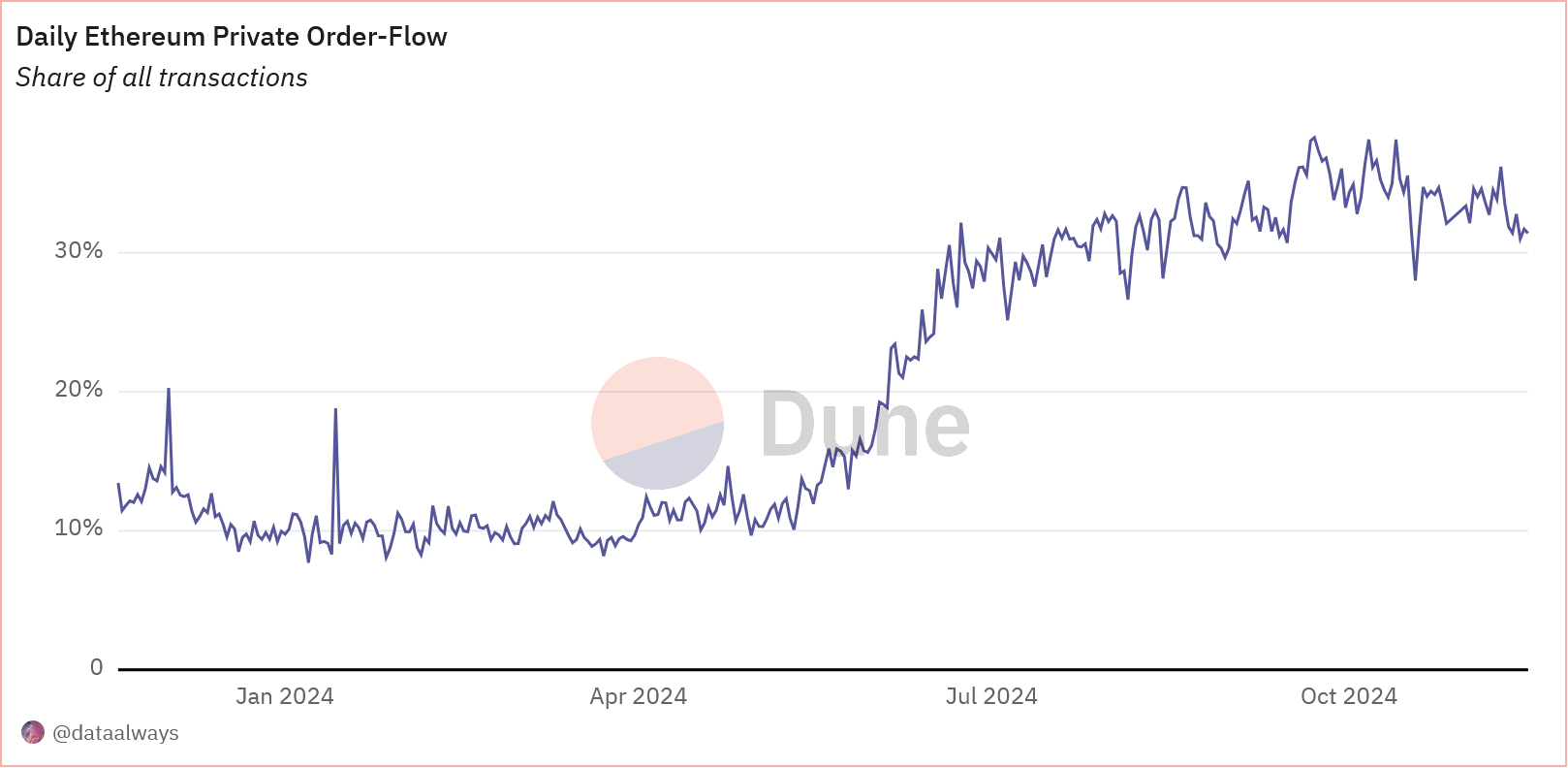}
    \caption{Private Transactions on Ethereum \protect\cite{dunePOF}}
    \label{fig::blocknative}
\end{figure}

\section{Related Work}
\label{sec:related}
%%%%%%%%%%%%%%%%%%%%%%%%%%%%%%%%%%%%%%%%

% %%%%%%%%%%%%%%%%%%%%%%%
% \subsection{MEV}
% \label{ssec:mev-rw}
% %%%%%%%%%%%%%%%%%%%%%%%
\textbf{MEV Auctions.}~\cite{10634354} study various strategic interactions and auction setups of block builders with proposers. They evaluate how access to MEV opportunities and improved relay connectivity impact bidding performance.~\cite{10271857} propose an Ethereum gas auction model using the First Price Sealed-Bid Auction (FPSBA) between different bots and miners.

\noindent\textbf{MEV redistributions.}\cite{chionas2023gets} model the MEV setting as a dynamical system with a fraction of MEV going to the miner as a dynamical variable updated with every time step. The miners and builders are assumed to be one entity, with the rest of the MEV returning to transaction creators. ~\cite{mazorra2023towards} discuss rebates in the context of liquidity providers in constant function market makers and discuss the auction between searcher and builders with the assumption of perfect MEV oracle that can compute the MEV extracted given the state of the blockchain and a new block of transactions.~\cite{chitra2022improving} proposes MEV redistribution as a dynamical system in which lending and staking portfolios of block proposer are chosen as a parameter that determines how much of the MEV extracted in a block is redistributed to staking.

\noindent\textbf{Game Theory and Blockchains} Researchers explored various game theoretic concepts in blockchains. E.g., the authors of \cite{roughgarden, damle2024designing,damle2024no}  use concepts from mechanism design to design transaction fee mechanisms and fairness. \cite{jain2021we,chen2024game} study scalability issues in blockchains through game theory.\cite{siddiqui2020bitcoinf} discusses on achieving fairness for Bitcoin in a transaction fee-only model. \cite{jain2022tiramisu} studies the equilibrium behavior of the miners. In this work, we explore the use of cooperative game theory in matchmaking. \cite{damle2021fasten} proposes a fair and secure distributed voting system that utilizes smart contracts to ensure that votes remain anonymous and are not tampered with during the voting process. \cite{faltings2021orthos} explores an AI framework designed for trustworthy data acquisition to enhance the reliability of data collection processes in multi-agent systems. Furthermore, \cite{srivastava2024decent} proposes a mechanism that promotes decentralization through block reward systems, aiming to enhance fairness and efficiency in blockchain networks.

%%%%%%%%%%%%%%%%%%%%%%%%%%%%%%%%%%%%%%%%
\section{Preliminaries}
\label{sec:prelims}
%%%%%%%%%%%%%%%%%%%%%%%%%%%%%%%%%%%%%%%%
% \subsection{MEV Opportunities}
% MEV opportunities arise due to market inefficiencies. Market inefficiencies generally occur when asset prices do not accurately reflect their true value, often due to fragmented liquidity across decentralized exchanges, information asymmetry, or flaws in protocol design, creating opportunities for arbitrage and other strategies. The price of execution is influenced by fluctuating gas fees, which increase during periods of high demand, and searchers often exploit this by optimizing gas prices or participating in priority gas auctions to have their transactions processed first. Additionally, supply shortages can cause significant price swings on decentralized exchanges, presenting further opportunities for arbitrage and market-making. Execution delays, such as the time lag between submitting a transaction and its inclusion in a block, can also be exploited by searchers who use techniques like 
%% front-running, where they place transactions with higher gas fees to execute before others, or latency arbitrage, where faster access to information allows them to capitalize on price changes before others react. Searchers try to capture these opportunities by strategic actions such as 
% front-running, back-running, sandwiching, etc~\cite{chi2024remeasuring},~\cite{wang2022impact}. Overall, searchers leverage these market inefficiencies, high execution costs, and transaction delays to execute various strategies to maximize profits in decentralized markets.

\subsection{Matchmaking} 
\emph{Matchmaking} executes in two steps: (i) the allocation of transactions to searchers and (ii) the redistribution of revenue generated through OFA. Consider a set of transactions $\mathcal{T} = \{t_1,t_2,\dots t_n\}$ that can generate MEV, and a set of searchers $\mathcal{S} = \{s_1,s_2,\dots s_m\}$, and a matchmaker $M$. 

Matchmaking aims to assign these transactions to searchers to maximize their welfare. Each searcher $s_i$ has valuation $v_{s_i}:2^{\mathcal{T}} \rightarrow \mathbb{R}_+ $. $M$ an auction among the searchers. Each searcher makes appropriate bids. Each transaction is allocated to one searcher based on the searcher's bids. Then, the matchmaker collects payments $p_{s_i}$ from each searcher $s_i$. The utility gained by searchers is due to the private transactions. Thus, transaction creators also play a significant role in maintaining the MEV ecosystem. In order to compensate them, the matchmaker shares the revenue generated by the auction, $\mathcal{R}=\sum_{i}{p_{s_i}}$, among the transaction creators.

\begin{definition}
[Matchmaking] We define matchmaking as mechanism $\mathcal{M}$ which takes, $(\mathcal{T}, \mathcal{S}, (v_{s_i})_{s_i \in \mathcal{S}})$ as inputs, (conducts auction amongst $\mathcal{S}$ for $\mathcal{T}$,) and outputs $(A,\Gamma)$ where $A$ represents allocation of $\mathcal{T}$ and $\Gamma$ represents the revenue distribution of $\mathcal{R}$. Reward to $t_j$ is given by $r_{t_j} =\Gamma_{t_j} \mathcal{R}$ and $\sum_{t_j\in\mathcal{T}} \Gamma_{t_j}=1$. 
\label{def:matchmaking}
\end{definition}

% \yc{In the following sections, we discuss different valuation models, how auctions are conducted, how payments are collected}

Before discussing matchmaking mechanisms, we first discuss the possible valuation structures and corresponding type of auction and its payments in the following sections.

\subsection{Searcher-Matchmaker Auctions}
The searcher-matchmaker auction is combinatorial in nature, with searchers competing for a subset of transactions. Technically, the valuations for different transactions and bundles may have complex relationships. However, commonly it is seen in two forms: one where the searcher values each transaction separately and another where the searchers bid for a bundle of transactions. The former can be seen as searchers with \emph{additive valuations} and the latter as searchers with \emph{single-minded valuations}. 
% \begin{table}[h]
%     \centering
%     \caption{Variations of Searcher-matchmaker Auction}
%     \begin{tabular}{|c|c|c|}
%          \hline
%          S.no & Agent Type $v_{s_i}$ & Bidding Language \\
%          \hline
%          1 & Additive & OR \\
%          2 & Single-Minded & Atomic\\
%          % 3 & SuperAdditive & XOR\\
%          \hline
%     \end{tabular}
%     \label{tab:sma}
% \end{table}

\noindent\textbf{Additive Valuations.} 
Searchers with additive valuations compete for the subset of transactions for which they have positive valuations. 
The valuation of any subset of items for a searcher is simply the sum of the individual valuations for the transactions in that subset, $v_{s_i}(B) = \sum_{t_j \in B}v_{s_i}(t_j)$. This model assumes that searchers value transactions independently of each other.  Each searcher $s_i \in \mathcal{S}$ submits the $n$-tuple $b_{s_i}$ 
% $B_{s_i} \subseteq \mathcal{T}$ in its OR bid $\{ \mathbf{OR} \big(t_j, b_{s_i}(t_j) \big)\  \lvert \ t_j \in B_{s_i}\}$. 

\begin{definition}[Additive Valuations]
    A valuation \( v \) is called \textit{additive} if \( v(B) = \sum_{i \in B} v(i) \), where \( v(i) \) is the value of transaction \( i \) and \( B \) is any subset of transactions.
\end{definition}

The optimal allocation in such auctions is to give each transaction to the searcher that values it the most. The optimal allocation in such auctions—where searchers have additive valuations and bid for individual transactions—can indeed be implemented using second-price auctions with Vickrey-Clarke-Groves (VCG) payments individually for each transaction.  VCG payments ensure desirable properties such as incentive compatibility, individual rationality, and optimal social welfare~\cite{narahari2014game}. 

\textbf{Single-Minded Valuations.}
To extract MEV, searchers often require the transactions in their bundles to be executed atomically, and hence, they only value a bundle if it is received in its entirety. Such searchers are referred to as single-minded searchers. A single-minded searcher values only one specific set or bundle of items, say $B^*$ and has no interest in any other proper subset of $B^*$.

\begin{definition}[Single-minded Valuations]
    A valuation $v$ is called \textit{single-minded} if there exists a bundle of transactions \( B \) and a value $\bar{v} \in \mathbb{R}^+$ such that \( v(B') = \bar{v} \) for all \( B' \supseteq B \), and \( v(B') = 0 \) for all other \( B' \). %A single-minded bid is the pair \( (B^*, v^*) \).\yc{single-minded bid may not be the best terminology. why need textit here and previous definition?}
\end{definition}

Finding the optimal allocation in this setting requires finding bundles that maximize social welfare while ensuring the bundles in optimal allocation are mutually disjoint. Both finding the optimal allocation and finding an optimal allocation that is better than $m^\frac{1}{2} - \epsilon$ is NP-hard ~\cite{blumrosen2007algorithmic}. Thus, such auctions are typically implemented using approximately social welfare auction mechanisms. In this paper, we use ~\cite{blumrosen2007algorithmic}'s ICA-SM, which greedily computes the best allocation and is $\sqrt{m}$-approximate optimal. ICA-SM ensures incentive compatibility and individual rationality and produces a $\sqrt{m}$-approximation of optimal social welfare. 

\begin{algorithm}
\caption{ICA-SM \protect\cite{blumrosen2007algorithmic}}
\label{algo:icasm}
\begin{algorithmic}[1]
\STATE $\mathbf{  Input:}$ Bids $\{(B_{s_i}, b_{s_i})\}_{i=1}^n$
\STATE R $\leftarrow$ \text{Sorted indices of searchers based on $b_{s_i} / \sqrt{\lvert B_{s_i}\rvert}$} 
\STATE $W \leftarrow \emptyset$.
\FOR{$k \in R$}
    \IF{$B_{s_k} \cap (\bigcup_{s_j \in W} B_{s_j}) = \emptyset$}
        \STATE $W \leftarrow W \cup \{B_{s_k}\}$.
    \ENDIF
\ENDFOR
\STATE $\mathbf{Payments:}$ For each $s_i \in W$, $p_{s_i} = b_{s_j} / \sqrt{\lvert B_{s_j}\rvert / \lvert B_{s_i}\rvert }$, where $j \in \{i+1, \ldots, n\}$ is the smallest index such that $B_{s_i} \cap B_{s_j} \neq \emptyset$, and for all $l < j, l \neq i$, $B_{s_l} \cap B_{s_j} = \emptyset$ (if no such $s_j$ exists then $p_{s_i} = 0$).
\STATE $\mathbf{  Output:}$ Set of Winners $W$

% $\mathbf{Allocation:}$ The set of winners is $W$. \\
\end{algorithmic}
\end{algorithm}

\subsection{Towards Matchmaking Mechanism}
% As stated in our introduction, there is no existing matchmaking mechanism. 
We address the matchmaking problem as follows.
The value generated by each transaction need not be the same for all transaction creators. Thus, it would be equitable for all transaction creators to redistribute the generated revenue proportional to their respective contribution to the system. This naturally leads us to the solution concept of \emph{Shapley} value. Towards this, we define a game known as \mevgame, and redistribute the revenue generated by the matchmaker according to the Shapley value of each transaction in this game.

%%%%%%%%%%%%%%%%%%%%%%%
\subsection{Characteristic Form Games and Shapley Value}
\label{ssec:cgt}
%%%%%%%%%%%%%%%%%%%%%%%

Cooperative game theory analyzes scenarios where players, or agents, can form coalitions to achieve collective goals. A \emph{characteristic form game} $(N,\nu)$ is (i) a set of players, and (ii) $\nu: 2^N \rightarrow \mathbb{R}$, $\nu(S)$ represents the value $S\subset N$ can generate by forming a coalition. For a cooperative game, many solutions are proposed to redistribute value collectively among participants (or players). 
% Shapley Value~\cite{winter2002shapley} is the most popular and widely used concept. It assigns a value to each player based on their contribution to the overall value generated by the coalition of players.

\noindent\textbf{Shapley Value}
Shapley value redistributes a cooperative game's total value or payoff to individual players based on their marginal contribution to every possible coalition. 
% It considers all possible orders in which players can join the coalition and calculates the average marginal contribution of each player. 
It is the only solution concept that satisfies all desirable properties of fair redistribution, such as efficiency, symmetry, and additivity~\cite{shapley1953value}. $\varphi_j(\nu)$ of  player $j$ playing in  characteristic form game $(N, \nu)$ is given as follows:

\begin{equation}
    {\varphi_{j}(\nu)} = \sum_{S\subseteq\mathcal{N}\setminus j} { \frac{\lvert S\rvert !(\lvert\mathrm{N}\rvert - \lvert S\rvert-1)!}{\lvert\mathrm{N}\rvert!}  \big(\nu(S \cup j)-\nu(S)\big)}\label{eq:shap-val-form-1}
\end{equation}
% 

% $\nu$ is the value function, where $\nu(S): \rightarrow \mathbb{R}$ represents the value or payoff of coalition \( S \subseteq N \) (the set of all players) and $\lvert N \rvert$ is the total number of players.
Alternatively, Shapley value can be computed via permutations of players \( \pi \in \Pi \), where $\Pi$ is the set of all permutations of players in which players could join the coalition. $\pi(j)$ represents the set of players that precede \( j \) in the permutation \( \pi \). The expression \( \nu(\pi(j) \cup \{j\}) - \nu(\pi(j)) \) is the marginal contribution of player \( j \) to the coalition $\pi(j)$.

\begin{equation}
\label{eq:shapley-form-2}
\varphi_j(\nu) = \frac{1}{\lvert N \rvert!} \sum_{\pi \in \Pi} \left[\nu(\pi(j) \cup \{j\}) - \nu(\pi(j))\right]    
\end{equation}

% \textbf{2. Banzhaf Value $\beta_j(\nu)$.}
% The Banzhaf value is another measure of a player’s influence in a cooperative game, specifically focusing on how often a player is pivotal in a coalition. It counts how many coalitions the player can ``swing" from a losing coalition (without them) to a winning one (with them). Unlike the Shapley value, which has a larger weight for larger coalitions, the Banzhaf value weighs each coalition equally.

% \begin{equation}
% \beta_{j}(\nu) =\frac{1}{2^{\lvert N \rvert-1}} \sum_{S \subseteq N} \nu(S) - \nu(S\setminus\{j\})        
% \end{equation}

% where \( \beta_j(\nu) \) is the Banzhaf value for player \( j \). 
% \( \nu(S) \) is the value of the coalition \( \mathcal{T} \) and \( \nu(S \setminus \{j\}) \) is the value of the coalition without player \( j \). The Banzhaf value measures how often player \( j \) is pivotal (or decisive) in changing a coalition’s outcome by being included or excluded.
% It does not require averaging over permutations like the Shapley value but instead considers each coalition where the player is pivotal.

% \textbf{3. Externality Value $\xi_j(\nu)$.} We define Externality value as another measure of a player's contribution to the system that captures the difference in value or payoff to the grand coalition with and without the player. 

% \begin{equation}
    % \xi_i(\nu) = \nu(N) - \nu(N \setminus \{j\})
% \end{equation}

%%%%%%%%%%%%%%%%%%%%%%%
\section{Our Approach}
\label{ssec:model}
%%%%%%%%%%%%%%%%%%%%%%%
\subsection{\mevgame}

\mevgame\ is a cooperative game ($\mathcal{T},\nu$) with $\mathcal{T}$, the transaction creators being the players. $\nu(T)$ where $T\subseteq \mathcal{T}$ is the value of transactions in $T$. We define it as the revenue collected by the matchmaker if only transactions $T$ had been available i.e., the scenario when only searchers whose desired bundle $\tau \subseteq T$ are present. For each $t_j$, its marginal contribution to each $T \subseteq \mathcal{T} \setminus t_j$, requires finding the revenue with $T \cup t_j$ and $T$. $\Gamma^{SHAP}_{t_j}$ is computed as $\frac{\varphi_{t_j}}{\sum_{i \in [n]} \varphi_{t_i}}$. 

For example, consider RST-Game with 3 searchers and 4 transactions, with corresponding (bundle, bid) pairs:  $(\{1,2\},10), (\{3,4\},9), (\{2,4\},8)$. In this case, the winners are $s_1, s_2$ and each of them pays is $(8/\sqrt(2))*\sqrt{2} = 8$ and the revenue is 16. $\Gamma_{t_1}^{SHAP} = \Gamma_{t_3}^{SHAP} = 0.154$, $\Gamma_{t_2}^{SHAP} = \Gamma_{t_4}^{SHAP} = 0.346$.

\subsection{Shapley Value of \mevgame\ with Additive Valuation Searchers}
%%%%%%%%%%%%%%%%%%%%%%%%%%%%%%%%%%%%%%%%%%%%%%

% We resort valuations to only additive, superadditive, and bidding language to atomic and OR bids, as we believe these closely model the real-time behavior and dynamics of searchers.

Searchers bid for the subset of transactions and valuation for each transaction. Each searcher $s_i \in \mathcal{S}$ submits bid
% $B_{s_i} \subseteq \mathcal{T}$ in its OR bid $\{ \mathbf{OR} \big(t_j, b_{s_i}(t_j) \big)\  \lvert \ t_j \in B_{s_i}\}$. 
$b_{s_i}$, where $b_{s_i} \in \mathbb{R}^n$ is an $n$-tuple where $b_{s_i}[j]$ is $s_i$ searcher's bid for transaction $t_j$. 

$M$ reduces this auction to $n$ independent second-price auctions with VCG payments. For each transaction, $t_j \in \mathcal{T}$, $M$ determines the searcher with the highest bid for $t_j$ as the winner. The winner pays the amount of the second-highest bid. The algorithm for the winner determination in this case can be found in appendix.
% Algorithm \ref{algo::spa} shows the winner determination in this case.

% We use the Shapley value to redistribute the revenue among the transactions $\mathcal{T}^{*}$. 
The Shapley value of each winning transaction $t_j \in \mathcal{T}^{*}$ is computed using the permutation method. Let $\Pi$ denote all the possible permutations of winning transaction set $\mathcal{T}^{*}$. For each permutation $\pi \in \Pi$, $\pi(t_j)$ denotes all the transactions present before $t_j$ in $\pi$. We compute Shapley value of $t_j$ as the average of the marginal contributions of $t_j$ for each permutation $\pi$ given by $\nu(\pi(t_j) \cup \{t_j\}) - \nu(\pi(t_j))$, where $\nu(\pi(t_j)) = \sum_{t_k \in \pi(t_j)} b_{s_\text{second}}[t_k]$. 

% Intuitively, this corresponds to the marginal contribution of $j$ in the revenue of each $T \subseteq \mathcal{T} \setminus \{j\}$.           $\nu(\pi(j) \cup \{j\}$ corresponds to the revenue generated when the only transactions being auctioned off are $\pi(j) \cup \{j\}$ (in coalition) by searcher $S^{\pi(j) \cup \{j\}}$ and $\nu(\pi(j)$ corresponds to the revenue generated when the only transactions being auctioned off are $\pi(j)$ ( not in coalition) by searchers $S^{\pi(j)}$. 

$$
\varphi_{t_j}(\nu) = \frac{1}{k!}\sum_{\pi \in \Pi} \nu(\pi(t_j) \cup \{t_j\}) - \nu(\pi(t_j))
$$

\begin{theorem}
\label{thm:easy}
 The Shapley value of \mevgame\ ($\mathcal{T},\nu$)) can be computed in polynomial time if $\nu$ is additive.
\end{theorem}
\begin{proof}
    As the searchers' valuations are additive, the marginal contribution of $t_j$ in each $k!$ permutation is exactly the payment of the highest bidder $b_{s_\text{second}}[t_j]$. Hence, the Shapley value of $t_j$ is $b_{s_\text{second}}[t_j]$ and redistribution fraction is $\Gamma^{\shap\ }_{t_j} = \frac{b_{s_\text{second}}[t_j]}{\sum_{{t_j} \in T^{*}} b_{s_\text{second}}[t_j]}$. 
\end{proof}
% where for any $T \subseteq \mathcal{T}, \nu(T) = \sum_{k \in T} b_{\text{second}}^k$. 

% Similarly, the redistribution fraction computed via the Banzhaf is the same as the Shapley value. 
% Banzhaf value $\beta_{t_j}(\nu)$ is equal to Externality $\xi_{t_j}(\nu)$ for all $j$ is

% $$
% \xi_{t_j}(\nu) = \nu(\mathcal{T}) - \nu(\mathcal{T} \setminus \{t_j\}) = b_{s_\text{second}}^{t_j}
% $$

% $$
% \beta_{t_j}(\nu) =\frac{1}{2^{n-1}} \sum_{T \subseteq \mathcal{T}} \nu(T \cup \{t_j\} ) - \nu(T) = b_{s_\text{second}}^{t_j}
% $$

% Since, $\forall t_j \in T^{*}$, $\xi_{t_j}(\nu) = \beta_{t_j}(\nu) = \varphi_{t_j}(\nu)$, redistribution fraction $rf^{ext}_{t_j} = rf^{banz}_{t_j} = rf^{shap}_{t_j}$. 

%%%%%%%%%%%%%%%%%%%%%%%
\subsection{Shapley Value of \mevgame\ for Single-Minded Searchers}
\label{ssec:shap-val-single-minded}
%%%%%%%%%%%%%%%%%%%%%%%
  Each searcher $s_i \in \mathcal{S}$ submits only a single subset $B_{s_i} \subseteq \mathcal{T}$ in bid $\{ B_{s_i}, b_{s_i} \}$, where $s_i$'s valuation $v_{s_i}$ is single-minded. $M$ reduces this auction to a combinatorial auction with single-minded bidders.

Let $\mathcal{T}^{*} \subseteq \mathcal{T}$ denote the winning transactions in the auction.
% Similar to Case 1, we define the characteristic form game $(\mathcal{T}, \tilde{\nu})$, such that for any $T \subseteq \mathcal{T}$ and for any $t_j \in \mathcal{T}$, $\tilde{\nu}(T)$ \yc{($T \subseteq \mathcal{T}\setminus\{t_j\}$), shouldn't $\tilde{\nu(T)}$ have $j$ as it's there on RHS?} is defined as follows: 
% \begin{equation}
%     \tilde{\nu}(T)=\begin{cases}
%         0 & \nu(T \cup \{t_j\}) - \nu(T) < 0 \\
%         \nu(T) & \text{otherwise}
%     \end{cases} 
% \end{equation}
% Due to the greedy approach, the marginal contribution $\nu(T\cup\{t_j\}) - \nu(T)$ is not always positive since for any $U,V$ such that $V \subseteq U \nRightarrow \nu(V) < \nu(U)$. Hence,  $\forall T \subseteq \mathcal{T}, \nu(T\cup\{t_j\}) < \nu(T)$, we consider zero marginal contribution. 
In most instances of the \mevgame\ generated at random, the number of unique marginal contributions seems polynomial. However, we show in Theorem \ref{thm:hard} that the number of unique values of the marginal contribution in \mevgame\ can be subexponential in certain instances due to the underlying (bundle, bid) structure.

\begin{theorem}
    \label{thm:hard}
    The number of unique marginal contributions in the computation of the Shapley value of transaction creators in \mevgame\  can be $\Omega(2^{\sqrt{n}})$.
\end{theorem}
\begin{proof}
Proof by construction:
%   Hello  
% \end{proof}

% \begin{example}[]\upshape
% Suppose that the Shapley value of \mevgame\ is computable in polynomial time, i.e., for all instances of $(\mathcal{T},\mathcal{S},\mathcal{B}, \nu)$ of \mevgame\, the computation of Shapley value of all winning transaction $\mathcal{T}^{*}$ is $\mathcal{O}(n^c)$.
Consider an instance of the \mevgame\ with set of transactions $\mathcal{T} = \{t_1, \ldots t_n\}$. Let the set of searchers $\mathcal{S} =\{s_1, \ldots s_m\}$  be divided into two classes of searchers with equal cardinality, $\mathcal{S}_0$ and $\mathcal{S}_1$ with corresponding set of bundles $\mathcal{B}_0$ and $\mathcal{B}_1$ satisfying the following $2$ properties:
\begin{itemize}
    \item None of the bundles in $\mathcal{B}_0$ or $\mathcal{B}_1$ intersects with any other bundle in $\mathcal{B}_0$ or $\mathcal{B}_1$, respectively, i.e., $\forall B_{s_i}, B_{s_j} \in \mathcal{B}_0, B_{s_i} \cap B_{s_j} = \phi$. In other words, bundles within $\mathcal{B}_0$ are mutually exclusive and bundles within $\mathcal{B}_1$ are mutually exclusive.
    \item Each bundle in $\mathcal{B}_0$ intersects with every bundle in $\mathcal{B}_1$ and vice versa, i.e., $\forall B_{s_i} \in \mathcal{B}_0,\forall B_{s_j} \in \mathcal{B}_1, B_{s_i} \cap B_{s_j} \neq \phi$.
\end{itemize}
%for each searcher $s_i$, it's interested bundle $B_{s_i}$ intersects with exactly $n -\sqrt{n}$ other bundles. 

Here is a construction satisfying the above properties:\\
Let $m = 2\sqrt{n}$. Thus, $|\mathcal{B}_0| = |\mathcal{B}_1| = m/2 = \sqrt{n}$. Let the size of each bundle be $\sqrt{n}$, i.e., $|B_{s_i}| = \sqrt{n}, \forall s_i \in \mathcal{S}$. Consider the arrangement of transactions as matrix $E$ shown below:

\[
E = \begin{bmatrix}
t_1 & t_2 & t_3 & \cdots & t_{\sqrt{n}} \\
t_{\sqrt{n}+1} & t_{\sqrt{n}+2} & t_{\sqrt{n}+3} & \cdots & t_{2\sqrt{n}} \\
t_{2\sqrt{n}+1} & t_{2\sqrt{n}+2} & t_{2\sqrt{n}+3} & \cdots & t_{3\sqrt{n}} \\
\vdots & \vdots & \vdots & \ddots & \vdots \\
t_{n - \sqrt{n} + 1} & t_{n - \sqrt{n} + 2} & t_{n - \sqrt{n} + 3} & \cdots & t_n
\end{bmatrix}
\]

% The transaction in each bundles $B_{s_i}$ along with respective bid $b_{s_i}$ is given by:

% \begin{align*}
% B_{s_i} =& \left\{E[(r + d \cdot j) \% \sqrt{n}][j]\ \mid 
%          r = i \% \sqrt{n} ,d = \left\lceil \frac{i}{\sqrt{n}}\right\rceil \right\}_{j \in \{1, \ldots, \sqrt{n}\}}\\
% b_{s_i} =&\ q^{r(i)}, q>2\\
%          & r(i) = \left( (i-1) \mod \sqrt{n} \right) \times \sqrt{n} + \left\lfloor \frac{i-1}{\sqrt{n}} \right\rfloor + 1\\
% \end{align*}

% Arranging the bundles into a matrix, we have 

% \[
% F = \begin{bmatrix}
% B_{s_1} & B_{s_2} & B_{s_3} & \cdots & B_{s_{\sqrt{n}}} \\
% B_{s_{\sqrt{n}+1}} & B_{s_{\sqrt{n}+2}} & B_{s_{\sqrt{n}+3}} & \cdots & B_{s_{2\sqrt{n}}} \\
% B_{s_{2\sqrt{n}+1}} & B_{s_{2\sqrt{n}+2}} & B_{s_{2\sqrt{n}+3}} & \cdots & B_{s_{3\sqrt{n}}} \\
% \vdots & \vdots & \vdots & \ddots & \vdots \\
% B_{s_{n - \sqrt{n} + 1}} & B_{s_{n - \sqrt{n} + 2}} & B_{s_{n - \sqrt{n} + 3}} & \cdots & B_{s_n}
% \end{bmatrix}
% \]
% % \[
% % h \quad \begin{array}{c@{\hspace{35pt}}c@{\hspace{30pt}}c@{\hspace{35pt}}c@{\hspace{30pt}}c}
% % \quad\quad\uparrow &  & \quad\quad\uparrow &  & \uparrow \\
% % \quad\quad1 &  & \quad\quad3 &  & \sqrt{n}
% % \end{array}
% % \]
% \[
% \quad 
% \begin{array}{c@{\hspace{20pt}}c@{\hspace{30pt}}c@{\hspace{28pt}}c@{\hspace{5pt}}c} 
% \textcolor{red}{\uparrow} &  & \textcolor{red}{\uparrow} &  & \textcolor{red}{\uparrow} \\
% \textcolor{red}{1} &  & \textcolor{red}{3} &  & \textcolor{red}{\sqrt{n}}\\
% \textcolor{red}{h =\{1,3,\sqrt{n}\}}& & & &
% \end{array}
% \]

\[
\mathcal{B}_0 = \left\{ 
\begin{array}{c}
\{t_1, t_2, \ldots, t_{\sqrt{n}}\}, \\
\{t_{\sqrt{n}+1}, t_{\sqrt{n}+2}, \ldots, t_{2\sqrt{n}}\}, \\
\ldots, \\
\{t_{n-\sqrt{n}+1}, t_{n-\sqrt{n}+2}, \ldots, t_n\} 
\end{array}
\right\}
\]

\[
\mathcal{B}_1 = \left\{ 
\begin{array}{c}
\{t_1, t_{\sqrt{n}+2}, t_{2\sqrt{n}+3} \ldots, t_{n}\}, \\
\{t_{\sqrt{n}+1}, t_{2\sqrt{n}+2}, \ldots, t_{n-1}, t_{\sqrt{n}}\}, \\
\ldots, \\
\{t_{n-\sqrt{n}+1}, t_{2}, t_{\sqrt{n}+3} \ldots, t_{n-\sqrt{n}}\} 
\end{array}
\right\}
\]

% $$\mathcal{B}_0 = \{ \{t_1, t_2, \ldots, t_{\sqrt{n}}\} \{t_{\sqrt{n}+1}, t_{\sqrt{n}+2}, \ldots, t_{2\sqrt{n}}\}, \ldots, \{t_{n-\sqrt{n}+1}, t_{n-\sqrt{n}+2}, \ldots, t_n\} \}$$

Bundles in $\mathcal{B}_0$ are essentially individual rows from the matrix $E$. Bundles in $\mathcal{B}_1$ contain one transaction from each row and each column of matrix $E$ \footnote{Note that $(\sqrt{n})!$ such selections of $\mathcal{B}_1$ are possible, given the same $\mathcal{B}_0$.
}. It can be shown that none of the bundles in $\mathcal{B}_0$ or $\mathcal{B}_1$ intersects with any other bundle in $\mathcal{B}_0$ or $\mathcal{B}_1$ respectively. It can also be shown that the intersection of a bundle from $\mathcal{B}_0$ and a bundle from $\mathcal{B}_1$ is a singleton \footnote{More specifically, $B^1_x \cap B^0_y$ is the transaction in $B^1_x \in \mathcal{B}_1$ with the row number $y$ in matrix $E$ where bundle $B^0_y \in \mathcal{B}_0$ is the bundle corresponding to row $y$ in matrix $E$.}.

Let us denote the bundles of $\mathcal{B}_0$ as $B^0_1, \ldots, B^0_{\sqrt{n}}$ and the corresponding bids as $b^0_1, b^0_2, \ldots, b^0_{\sqrt{n}}$ such that 
%$b^0_i > b^0_j$ iff $i < j$, i.e., 
$b^0_1 > b^0_2 > \ldots > b^0_{\sqrt{n}}$. 
Similarly, we denote the bundles of $\mathcal{B}_1$ as $B^1_1, \ldots, B^1_{\sqrt{n}}$ and the corresponding bids as $b^1_1, b^1_2, \ldots, b^1_{\sqrt{n}}$ such that 
% $b^1_i > b^1_j$ iff $i < j$, i.e.,  
$b^1_1 > b^1_2 > \ldots > b^1_{\sqrt{n}}$. For ease of exposition, we also require the bid values to be $3^{(2\sqrt{n}-2i+1)}$ for $b^0_i$ and $3^{(2\sqrt{n}-2i)}$ for $b^1_i$. Thus,
\begin{multline*}    
    b^0_1 = 3^{(2\sqrt{n}-1)} > b^1_1 = 3^{(2\sqrt{n}-2)} > b^0_2 = 3^{(2\sqrt{n}-3)} > \\ \ldots > b^0_{\sqrt{n}} = 3^1 >b^1_{\sqrt{n}} = 3^0 
\end{multline*}

% Let $F[k]$ denote the set of bundles in row $k$ of F. Note that $\forall B_{s_i}, B_{s_j} \in F[k], B_{s_i} \cap B_{s_j} = \emptyset$, and $\forall B_{s_l} \notin F[k], B_{s_i} \cap B_{s_l} \ne \emptyset$. There are no intersections for any two bundles in the same row, while there are intersections for bundles in different rows. Thus, every bundle has intersections with exactly $n- \sqrt{n}$ bundles. 

% Let $H = \mathcal{P}(\{1, 2, \ldots ,\sqrt{n}-1\}) \setminus \emptyset$, denote the indices (1 indexed) of bundles in $F[1]$. For any $h \in H$, $T_h = \bigcup_{l \in h} B_{s_l} \cup B_{s_{\sqrt{n} +l}}$ denote the union of all transaction $s_l$ and $s_{\sqrt{n}+l}$ are interested and $S(T_h) = \{B_{s_i}\mid\ B_{s_i} \subseteq T_h\}$ is set of all searchers who will be playing \mevgame\ when the only available transactions are $T_h$. Let $W(T_h)$ be the set of corresponding winning searchers and $P(T_h) = \{p^{T_h}_{s_i}\}$ be set of payments by winning searchers.
% % Observe that given $T_h$, $S^{T_h} = S_h$, which implies that searchers in $S_h$ are the only players in \mevgame\ if $h$ were the only transactions available. 

% The winning transactions $\mathcal{T}^{*}$ in this scenario is $\mathcal{T}$ and the corresponding winning searchers are $\{s_1, \ldots, s_{\sqrt{n}}\}$. Therefore, the matchmaker has to compute the Shapley value of each transaction. We now prove that the number of unique positive marginal contributions the matchmaker has to compute is at least subexponential in $n$; hence, the Shapley value computation is in subexponential.

We now show that while computing Shapley values of all transactions using ICA-SM
for the above construction, we will encounter at least $2^{\sqrt{n}}-1$ unique marginal values. The main insight in our approach is that corresponding to any of the non-empty $2^{\sqrt{n}} - 1$ combinations of bundles from $\mathcal{B}_0$, there exists a unique transaction set $T$ and transaction $t_j$ such that $\nu(T) - \nu(T \setminus \{t_j\})$ gives a unique marginal contribution.

\emph{Transaction selection}: Let the bundles selected from $\mathcal{B}_0$ in the combination be $B^0_a, \ldots, B^0_z$. Then, we define our transaction set as $T = \bigcup_{x=a}^{x=z}B^0_x \cup B^1_x$. We note the following:
\begin{enumerate}
    % \item Different selection of bundles gives a different $T$.\\
    % Case 1: All bundles from $\mathcal{B}_0$ are selected. Thus, $T = \mathcal{T}$, i.e., the set of all transactions.\\
    % Case 2: Let $B^0_x \in \mathcal{B}_0$ such that $B^0_x$ is not selected. In this case, we do not include $B^1_x$ in the union of transactions either. We know that $\exists t_i$ such that $t_i \in B^0_x$ and $t_i \in B^1_x$. As $t_i \in B^0_x, \forall B^0_y \in \mathcal{B}_0 \setminus \{B^0_x\},  t_i \notin B^0_y$. Similarly, $\forall B^1_y \in \mathcal{B}_1 \setminus \{B^1_x\},  t_i \notin B^1_y$. Thus, $t_i \notin T$. For all unselected bundles $B^0_x$, there exists a unique $t_i$. Thus, for all possible selections of bundles, we can find
    \item Let $\{t_x\} = B^0_x \cap B^1_x$ and $\{t_y\} = B^0_y \cap B^1_y$. Then, if $B^0_x \neq B^0_y \implies t_x \neq t_y$.\\
    We can prove this by contradiction. Suppose
    \begin{align*}
        \exists B^0_x, B^0_y | B^0_x \neq B^0_y \land t_x = t_y \\
        \implies& t_x \in B^0_x \land t_x \in B^0_y \\
        \implies& t_x \in B^0_x \cap B^0_y \\
        \implies& B^0_x \cap B^0_y \neq \phi
    \end{align*}
    Given that all bundles in $\mathcal{B}_0$ are mutually exclusive, this is a contradiction.
    \item Let $\{t_x\} = B^0_x \cap B^1_x$. If $B^0_x$ is selected, $t_x \in T$.\\
    This is trivially true because $T = \bigcup B^0_x \cup B^1_x$ for all selected bundles $B^0_x$.
    \item Let $\{t_x\} = B^0_x \cap B^1_x$. If $B^0_x$ is not selected, $t_x \notin T$.\\
    We prove this by contradiction. Let $B^0_y$ and $B^1_y$ represent any selected bundle. Suppose $t_x \in T$
    
    \begin{align*}
        &\implies (\exists B^0_y | t_x \in B^0_y) \lor (\exists B^1_y | t_x \in B^1_y)\\
        &\implies (t_x \in B^0_x \land \exists B^0_y | t_x \in B^0_y) \lor\\
        & \hspace{25pt} (t_x \in B^1_x \land \exists B^1_y | t_x \in B^1_y)\\
        &\implies (\exists B^0_y | t_x \in B^0_y \cap B^0_x) \lor (\exists B^1_y | t_x \in B^1_y \cap B^1_x) \\
        &\implies (\exists B^0_y | B^0_y \cap B^0_x \neq \phi) \lor (\exists B^1_y | B^1_y \cap B^1_x \neq \phi)
    \end{align*}
    Given that all bundles within $\mathcal{B}_0$ and $\mathcal{B}_1$ are mutually exclusive, this is a contradiction.
    \item While aggregating which searchers are allotted their bundles given the set of transactions $T$, only the following two cases are possible:
    \begin{itemize}
        \item All allotted bundles are from $\mathcal{B}_0$ and their payments are some bids $b^1_x$, i.e. bid value of a bundle in $\mathcal{B}_1$.
        \item All allotted bundles are from $\mathcal{B}_1$ and their payments are some bids $b^0_x$, i.e. bid value of a bundle in $\mathcal{B}_0$.
    \end{itemize}
    This is true because our construction obeys the two properties mentioned above: i) mutual exclusion within $\mathcal{B}_0$ and $\mathcal{B}_1$, and, ii) non-empty intersection between any bundle from $\mathcal{B}_0$ and $\mathcal{B}_1$.
    \item The ICA-SM payment
        \footnote{As the size of all bundles in $\mathcal{B}_0$ and $\mathcal{B}_1$ is $\sqrt{n}$, ICA-SM ordering is the same as ordering of the bids -- the common factor $\sqrt{n}$ cancels out. Similarly, the payments (Step 9 in Algorithm~\ref{algo:icasm}) $p_{s_i} = b_{s_j}/\sqrt{|B_{s_j}|/|B_{s_i}|} = b_{s_j}$ as $|B_{s_j}| = |B_{s_i}| = \sqrt{n}$}
    corresponding to $B^0_x$ is $b^1_x$.\\
    Transactions of both $B^0_x$ and $B^1_x$ are present in $T$. Also, according to our assigned bid values, $\nexists b_y | b^0_x < b_y < b^1_x$, i.e., bid $b^1_x$ is the highest bid smaller than bid $b^0_x$. Given that $B^0_x$ and $B^1_x$ intersect, the searcher corresponding to bundle $B^0_x$ pays $b^1_x$.
    \item The ICA-SM payment corresponding to the bundle $B^1_x$ is either $0$ (if $B^0_x$ was the bundle with the least bid among the selected bundles) or $b^0_y$ for some bundle $B^0_y$, where $b^0_y$ is the highest bid less than $b^1_x$.
    \item Let $B^0_a$ be the selected bundle with the highest bid. Let $B^0_w$ be one of the unselected bundles (if all bundles are selected, consider the least valued bundle $B^0_z$ instead). Thus, $\exists t_u \in B^0_a \cap B^1_w$.\\
    When the transaction set is $\tau = T \setminus \{t_u\}$, $\nu(\tau \cup \{t_u\}) = \sum_x b^1_x$ and $\nu(\tau) = (\sum_x b^0_x) - b^0_a$, where $B^0_a$ is the bundle with the highest bid.\\
    Thus, the marginal contribution of coalition $\tau$ for Shapley value of $t_u$ is $\sum_x b^1_x - \sum_x b^0_x + b^0_a$.
\end{enumerate}
Using (1), (2) and (3), it can be argued that corresponding to every bundle $B^0_x$, there is a unique transaction $t_x$ which belongs to $T$ if and only if $B^0_x$ is selected. Thus, while considering all non-empty selections of bundles, we are considering $2^{\sqrt{n}} - 1$ unique subsets of $\mathcal{T}$. Each of these unique sets of transactions is involved in marginal contributions of some transaction $t_u$, as shown in (7). The marginal contributions $\sum_x b^1_x - \sum_x b^0_x + b^0_a$ are the sum of all bids with coefficients either $-1, 0$ or $1$. Thus, for these bid values (powers of $3$) and bundles and coefficients $-1, 0$, or $1$, all possible summations, and hence the marginal contributions are unique. Refer to the Appendix for the proof.

% Thus, we proved that there exist instances of \mevgame\ in which the number of unique marginal contributions encountered is more than $(2^{\sqrt{n}} - 1)$, or of the order $\Omega(2^{\sqrt{n}})$ (i.e., \texttt{SUBEXP} in $n$), where $n$ is the number of transactions. Alternatively, the number of unique marginal contributions is more than $(2^{m/2}-1)$, i.e. $\Omega(2^{m/2})$ or in \texttt{EXP}-complexity class in terms of number of searchers $m$.

Thus, we proved that there exist instances of \mevgame\ in which the number of unique marginal contributions:
\begin{itemize}
    \item in terms of $n$: is more than $(2^{\sqrt{n}}-1)$. Thus, time complexity is $\Omega(2^{\sqrt{n}})$, i.e., \texttt{SUBEXP} in $n$.
    \item in terms $m$: is more than $(2^{m/2}-1)$. Thus, time complexity is $\Omega(2^{m/2})$, i.e., \texttt{EXP} in $m$.
\end{itemize}
\end{proof}

Games in which Shapley values can be computed in polynomial time often compute unique marginal contributions and find out how many times they occur. Both these steps must take polynomial time. In addition to these requirements, such games have a more well-defined structure that allows for optimizing the computation of Shapley value. \mevgame\ allows searchers to propose any arbitrary bundles. Any proposed algorithm must work for all such possible arbitrary bundles. With these two insights, we conjecture the following:
% As Shapley value computation require using all unique marginal contributions and there is no pattern among the marginal contributions, we believe precise computation of the Shapley value is at least SUBEXP. Hence we conjecture the following:

\begin{conjecture}
\label{con:hard}
    Shapley value computation of transaction creators in \mevgame\ with single-minded searchers can be SUBEXP in number of transaction creators $n$.
\end{conjecture}

\noindent\textbf{Unanimity Games and Shapley Value.} The structure of \mevgame\ can be derived from unanimity games. Let $U_S = (\mathcal{T},\omega_S)$ be unanimity game such that
\begin{equation}\forall T \subseteq \mathcal{T}, \omega_S(T) = 
\begin{cases}
    1\quad S\subseteq T\\
    0\quad \text{Otherwise}
\end{cases}
\end{equation}
The Shapley value of $t_j \in \mathcal{T}$ in $U_S$ is given by $\varphi_{t_j} = \frac{1}{\lvert S \rvert}$. The set of all unianimity functions $W = \{\omega_S| S\subseteq T\}$ forms a basis for the characteristic function $\nu:2^{\mathcal{T}}\rightarrow \mathbb{R}$ and $\nu$ can be represented as a linear combination of elements of $W$, i.e., $\forall C\subseteq T\setminus\phi, \nu(C) = \sum_{T \in 2^{\mathcal{T}}\setminus\phi} \Delta_T w_T(C)$, where $\Delta_T$s are referred to as \emph{Harsanyi dividends}. Based on unianimity game, the Shapley value can be written as 
\begin{equation}
    \varphi_{t_j} = \sum_{T\in 2^{\mathcal{T}}\setminus\phi,\ t_j \in T} \frac{\Delta_T}{\lvert T \rvert}
\end{equation}

The Harsanyi dividends $\Delta_T$ can be computed using the following equation \begin{equation}    
\Delta_T = \sum_{C\subseteq\mathcal{T}}(-1)^{\lvert \mathcal{T} \rvert-\lvert C \rvert} \nu(C)\end{equation} Note that $\Delta_T$ can be computed efficiently when $\lvert \{\nu(C) | C \in 2^{\mathcal{T}}\setminus \phi\}\rvert$ is polynomial in $n$ and $m$ i.e., the number of possible different revenues generated over all subset of searchers are polynomial.
With the construction provided above, one can observe that $\forall x \in\{1,\ldots\sqrt{n}\}, T_x = B^0_x \cup B^1_x$ the value of $\nu(T_x)$ is unique and thus, the number of distinct possible revenues can be sub-exponential. Hence, computing the dividends requires finding all the values. Therefore, the Shapley value computation via Unamity games would be sub-exponential.  Further, one can empirically verify Theorem~\ref{thm:hard} and observe that variation in the number of marginal contributions with an increasing number of transactions is sub-exponential. The details of our analysis are provided in the Appendix. To this end, we propose to use a randomized approach to approximate the Shapley value. 

\subsection{Approximating Shapley Value}
The Shapley value computation can go sub-exponential due to the underlying structure of the game, as shown in our construction. The occurrence of such structures is typically rare as the real world closely follows some distribution \footnote{We often see some transactions being more lucrative than others to almost all of the searchers and occasionally, some transactions being relatively highly valued by only a few (specialized) searchers~\cite{eigenphiDataEigenPhiWisdom}}. Due to this behavior (the number of different bundles submitted $\ll$ the total number of all possible bundles), the number of unique marginal contributions in Shapley value computation is $\mathcal{O}(n)$ (since most of them would be 0). So, we propose \ouralgo\, a randomized algorithm to compute the approximate Shapley value of each transaction. Algorithm~\ref{algo:rysp} describes \ouralgo. Let $\Pi$ be the set of all permutations of transactions $\Pi$. Let $\bar{\Pi}$ the set of $k$ different permutations sample from $\Pi$. Let $\bar{\Pi}_k$ denote the set of permutations sampled from $\Pi$. For each transaction $t_j$, the approximate Shapley value $\tilde{\varphi_j}$ is computed marginal contribution of ${t_j}$ to each $\pi \in \Pi$, averaged over $k$. Among the winning transactions selected via greedy approximation, the fraction of revenue redistributed to transaction creator $j$ is given by $\Gamma_{t_j}^{RSYP} = \frac{\tilde{\varphi}_{t_j}(\nu)}{\sum_{j\in[n]}\tilde{\varphi}_{t_j}(\nu)}$ using \ouralgo. We empirically show, for $k=\mathcal{O}(n^2)$, $\forall t_j \in \mathcal{T},  \Gamma_{t_j}^{\ouralgo}$ computed via \ouralgo\ approaches $\varphi_{t_j}(\nu)$. 

% With the characteristic form defined, across different valuation and bundles redistributions, we compute the importance of each transaction $t_j \in \mathcal{T}$ via Shapley value ($\tilde{\varphi}_{t_j}$), Banzhaf value ($\tilde{\beta}_{t_j}$), and Externality ($\tilde{Ext}_{t_j}$) with greedy allocation. As shown in Fig. \ref{fig::ratio}, we observe that the expected ratio of $\forall t_j \in \mathcal{T}, 
% \frac{\tilde{\varphi}_{t_j}}{\phi_{t_j}} < \frac{\tilde{\beta}_{t_j}}{\phi_{t_j}} < \frac{\tilde{Ext_{t_j}}}{\phi_{t_j}}$. Hence, we compute the redistribution fraction via $\tilde{\varphi}_{t_j}$.

% For the redistribution computed via a mechanism $X$, we want $\forall i\in \mathcal{T}, rf^{X}_{t_j} \approx rf^{shap}_{t_j}$.

\begin{algorithm}
\caption{\ouralgo}
\begin{algorithmic}[1]
\STATE $\mathbf{Input:}$ $\bar{\Pi}$, $n$, $k$
\FOR{$j =1$ to $n$}
    \STATE $MC_{sum}$ = 0
    \FOR{$\pi \in \bar{\Pi}$}
    \STATE $MC$ = $\nu(\pi(j) \cup j) - \nu(\pi(j))$
    \STATE $MC_{sum}$ += $MC$
    \ENDFOR
    \STATE $\varphi_{t_j}(\nu) = \frac{MC_{sum}}{k}$
\ENDFOR
\FOR{$j =1$ to $n$}
\STATE $\Gamma_{t_j}^{RSYP} = \frac{\tilde{\varphi}_{t_j}(\nu)}{\sum_{j\in[n]}\tilde{\varphi}_{t_j}(\nu)}$
\ENDFOR
\STATE $\mathbf{Output:} \{\ \Gamma_{t_j}^{RSYP}\}_{j\in[n]}$
\end{algorithmic}
\label{algo:rysp}
\end{algorithm}

\subsubsection{Error Analysis}
The computation of the Shapley value of any transaction $t_j$ is essentially a sum of marginal contributions weighted by $\frac{1}{|N!|}$ as shown in Equation \ref{eq:shapley-form-2}. A marginal contribution of $t_j$ to $T \subseteq \mathcal{T}$ in \mevgame\ is the difference in revenue (sum of payments) when only the searchers $s_i: B_{s_i} \subseteq T \cup \{t_j\}$  are present and $\{s_i: B_{s_i} \subseteq T\}$ are present. Thus, the marginal contribution in $\varphi_j$ computation is at most the maximum revenue and at least the minimum revenue.

Let \(\Pi^{\text{ord}}\) be set of all elements of \(\Pi\), ordered with a some total order on \(\Pi\). Let $\pi_i$ represent the permutation at $i^{th}$ index in $\Pi^{\text{ord}}$. Let $X_i^j$ represent the marginal contribution of $t_j$ to $\pi_i(j)$. Then $\varphi_{t_j}$ can be written as $\frac{1}{N!}\sum_{i \in [N!]} X_i^j $. Let $K$ represent the indices of permutations drawn uniformly at random from $\Pi^{\text{ord}}$, then we have, $\tilde{\varphi}_{t_j} = \frac{1}{k}\sum_{k\in K} X_k^j$.

We want to find out the following $P(\tilde{\varphi}_{t_j} - \varphi_{t_j} \geq t) \le 1 - \delta$. Since $X_i^j$ are bounded, we use Hoeffding's inequality to find the upper bound on  $P(\tilde{\varphi}_{t_j} - \varphi_{t_j} \geq t)$. From Hoeffding's inequality, we know that, for $S_n = \sum_{i \in [n]} X_i$, where $X_i$ are independent random variables and $X_i \in [a_i,b_i]$,
\begin{equation}
    P(S_n - \mathbb{E}[S_n] \ge t) \le e^{\frac{-2t^2n^2}{\sum_{i \in [n]} (b_i-a_i)^2}}
    \label{eq:hoeffding}
\end{equation}

$\mathbb{E}[\tilde{\varphi}_{t_j}] = \mathbb{E}\big[\frac{1}{k}\sum_{k\in K}X_k^j \big] 
% = \frac{1}{k} \sum_{k \in K} \mathbb{E}[X_k^j]
= \frac{1}{N!}\sum_{i \in [N!]} X_i^j =\varphi_{t_j}$. Thus, $P(\tilde{\varphi}_j - \varphi_j \geq t) = P(\tilde{\varphi}_j - \mathbb{E}[\tilde{\varphi}_j] \ge t)$. Thus, we have,

\begin{align}
    P(\tilde{\varphi}_j - \varphi_j \geq t) \le e^{\frac{-2t^2k^2}{\sum_{k \in K} (b-a)^2}} = e^{\frac{-2t^2k}{(b-a)^2}}
\end{align}

Let the sum of all bids $\sum_{i\in m}b_i = R^*$. Since, the revenue is at most $R^*$, we have $\forall i\in[N!], -R^* \le X_i^j \le R^*$, 
% $P(\tilde{\varphi}_j - \varphi_j \geq t) \le e^{\frac{-t^2k}{2{R^*}^2}}$. 
Using this, we get $k \ge \frac{2{R^*}^2}{t^2} \ln(\frac{1}{1 -\delta})$.

%%%%%%%%%%%%%%%%%%%%%%%%%%%%%%%%%%%%%%%%%%%%%%
\section{Experimental Analysis}
%%%%%%%%%%%%%%%%%%%%%%%%%%%%%%%%%%%%%%%%%%%%%%
\subsection{Setup}
% To empirically prove our claim that the number of unique marginal contributions is \(\Omega(2^{\sqrt{n}})\), we generate instances of \mevgame\ with the bundles and valuations discussed in the previous section \ref{ssec:shap-val-single-minded}. We show the variation of the number of unique marginal contributions with increasing transactions. 

We demonstrate the efficacy of \ouralgo\ in \mevgame\ by comparing the redistributions of Shapley values from the exact and approximation methods across 10K randomly generated \mevgame\ instances, each with varying searcher bids and transaction bundles. 
% For each instance, we compute compare the revenues fraction with exact Shapley value and \ouralgo.
% compare them to the approximate redistribution fraction obtained through our sampling method.
% The mean of fraction of revenue received by each over 10K random instance are close in both the cases. 
% (i) $D_1$: $v_i\sim\mathcal{U}(0,c)$ and  (ii) $D_2$: $v_i\sim\mathcal{N}(c,1)$  (iii) $D_3$: $v_i\sim\text{EXP}(c/2)$  (iv) $D_4$: $v_i\sim\text{Triangular}(0, c/2, c)$. We show that used, $c=10000$ for empirical analysis.

\subsection{Results} 

% Figure~\ref{fig:shap-exact-approx-comparison} shows the distribution of $\Gamma$ over 10K instances for $n=6$ and $m=6$ for each transaction creator.
We observe the mean redistribution fraction via exact Shapley value and \ouralgo\ are almost the same. Figures~\ref{fig:varyT} and~\ref{fig:varyM} show similar behavior over all the transaction creators for varying $n$ and $m$.
% Further, we observe that as the sample size $k$ increases, the ratio of $\tilde{\varphi}_j \text{and}\ \varphi_j$ decreases, .i.e.,$$\forall j \in [n], \lim_{k \to \infty} \frac{\tilde{\varphi}_j^{k}}{\varphi_j^k} \approx 1$$
Hence, $\Gamma_{\texttt{RSYP}} \rightarrow \Gamma_{t_j}^{SHAP}$ in $\texttt{RSYP}$ with $k=\mathcal{O}(n^2)$.

% \begin{figure}[ht]
%     \centering    
%     \includegraphics[width=\linewidth]{AAMAS25/Figures/Shap Distributions.pdf}
%     \caption{$\Gamma^{\ouralgo\ },\Gamma^{\shap\ }$ for $n$=6,\ $m$=6}
%     \label{fig:shap-exact-approx-comparison}
% \end{figure}

\begin{figure}[H]
    \centering
    \includegraphics[width=.8\linewidth]{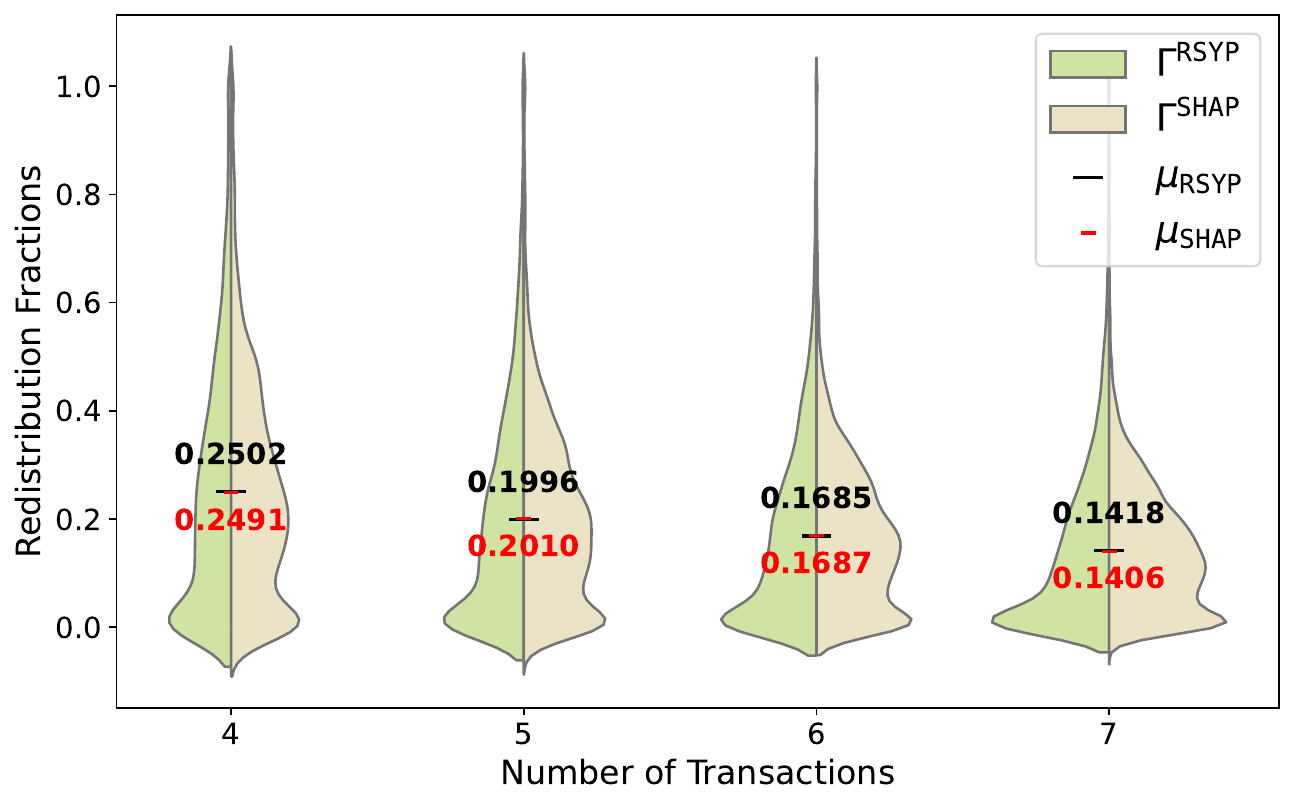}
    \caption{Distribution of $\Gamma^{\ouralgo\ },\Gamma^{\shap\ }$ vs $n$ for $m=6$}
    \label{fig:varyT}
\end{figure}

\begin{figure}[ht]
    \centering
    \includegraphics[width=0.8\linewidth]{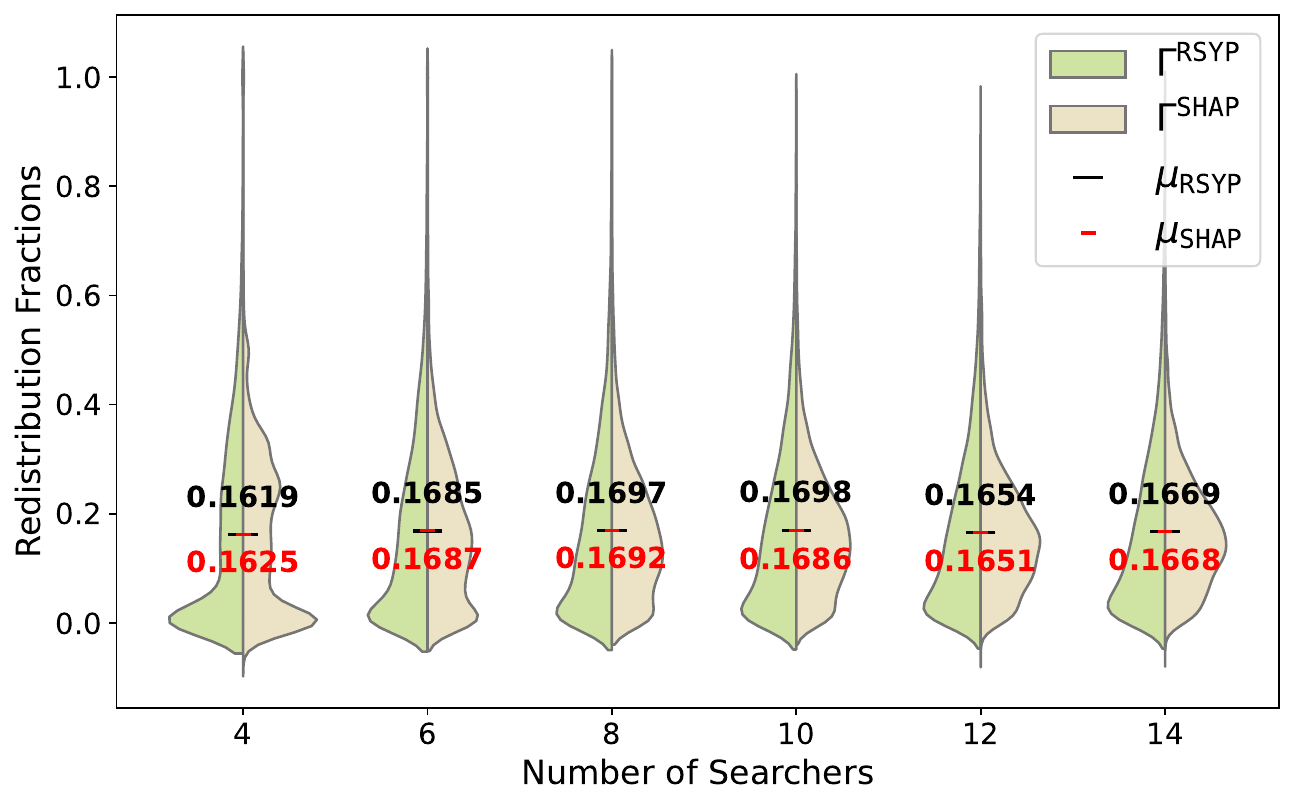}
    \caption{Distribution of $\Gamma^{\ouralgo\ },\Gamma^{\shap\ }$ vs $m$ for $n=6$}
    \label{fig:varyM}
\end{figure}

%%%%%%%%%%%%%%%%%%%%%%%%%%%%%%%%%%%%%%%%%%%%%%
\section{Conclusion}
%%%%%%%%%%%%%%%%%%%%%%%%%%%%%%%%%%%%%%%%%%%%%%
In this work, we explored the problem of matchmaking in MEV. We defined a cooperative game \mevgame\ over transaction creators and proved that computing Shapley value for fair revenue redistribution among transaction creators is SUBEXP. We proposed a randomized algorithm that approximates the Shapley value very well for $\mathcal{O}(n^2)$ where $n$ is the number of transactions.

\noindent\textbf{Future Work.} We leave the future to analyze the complexity of the fair revenue redistribution for more general valuations of searchers. We believe the complexity would be at least as much as in the single-minded case. 
% Further, one could determine the theoretical bounds of the efficacy of the randomized algorithm proposed here. 

% Future directions in matchmaking can explore matchmaking with different objectives, such as Nash social welfare. Further, one could look at incentives for matchmakers or making the process credible. Another direction would be designing decentralized matchmaking without a matchmaker with an incentive structure for searchers and transaction creators. 
\bibliographystyle{unsrt}  
\bibliography{references}

\begin{thebibliography}{10}

\bibitem{swan2015blockchain}
Melanie Swan.
\newblock {\em Blockchain: Blueprint for a new economy}.
\newblock " O'Reilly Media, Inc.", 2015.

\bibitem{nakamoto2008bitcoin}
Satoshi Nakamoto.
\newblock Bitcoin: A peer-to-peer electronic cash system.
\newblock {\em Satoshi Nakamoto}, 2008.

\bibitem{buterin2013ethereum}
Vitalik Buterin et~al.
\newblock Ethereum white paper.
\newblock {\em GitHub repository}, 1:22--23, 2013.

\bibitem{ethereumMaximalExtractable}
Ethereum.
\newblock {M}aximal extractable value ({M}{E}{V}) | ethereum.org --- ethereum.org.
\newblock \url{https://ethereum.org/en/developers/docs/mev/}.
\newblock [Accessed 23-09-2024].

\bibitem{chainMaximalExtractable}
Chainlink.
\newblock {M}aximal {E}xtractable {V}alue ({M}{E}{V}) | {C}hainlink --- chain.link.
\newblock \url{https://chain.link/education-hub/maximal-extractable-value-mev}.
\newblock [Accessed 23-09-2024].

\bibitem{burian2024futuremev}
Jonah Burian.
\newblock The future of mev, 2024.

\bibitem{torres2024rollingshadowsanalyzingextraction}
Christof~Ferreira Torres, Albin Mamuti, Ben Weintraub, Cristina Nita-Rotaru, and Shweta Shinde.
\newblock Rolling in the shadows: Analyzing the extraction of mev across layer-2 rollups, 2024.

\bibitem{luganodesLuganodesWhat}
Luganodes.
\newblock {L}uganodes | {W}hat is {M}{E}{V}?: {E}xplained --- luganodes.com.
\newblock \url{https://www.luganodes.com/blog/mev-explained-maximal-extractable-value/}, 2023.
\newblock [Accessed 1-10-2024].

\bibitem{flashbots}
Philip Daian, Steven Goldfeder, Tyler Kell, Yunqi Li, Xueyuan Zhao, Iddo Bentov, Lorenz Breidenbach, and Ari Juels.
\newblock Flash boys 2.0: Frontrunning in decentralized exchanges, miner extractable value, and consensus instability.
\newblock In {\em 2020 IEEE Symposium on Security and Privacy (SP)}, pages 910--927, 2020.

\bibitem{ramosMev}
Simona Ramos and Joshua Ellul.
\newblock The mev saga: Can regulation illuminate the dark forest?
\newblock In Marcela Ruiz and Pnina Soffer, editors, {\em Advanced Information Systems Engineering Workshops}, pages 186--196, Cham, 2023. Springer International Publishing.

\bibitem{ji2024regulatory}
Yan Ji and James Grimmelmann.
\newblock Regulatory implications of mev mitigations.
\newblock In {\em Proceedings of the 5th Workshop on the Coordination of Decentralized Finance}, 2024.

\bibitem{chaurasia2024mev}
Yash Chaurasia, Parth Desai, Sujit Gujar, et~al.
\newblock Mev ecosystem evolution from ethereum 1.0.
\newblock {\em arXiv preprint arXiv:2406.13585}, 2024.

\bibitem{gupta2023centralizingeffectsprivateorder}
Tivas Gupta, Mallesh~M Pai, and Max Resnick.
\newblock The centralizing effects of private order flow on proposer-builder separation, 2023.

\bibitem{frontierBuilderDominance}
Frontier Research.
\newblock {B}uilder {D}ominance and {S}earcher {D}ependence --- frontier.tech.
\newblock \url{https://frontier.tech/builder-dominance-and-searcher-dependence}, 2023.
\newblock [Accessed 22-09-2024].

\bibitem{flashbotsSearchingPostMerge}
https://twitter.com/zeroXbrock.
\newblock {S}earching {P}ost-{M}erge | {F}lashbots {W}ritings --- writings.flashbots.net.
\newblock \url{https://writings.flashbots.net/searching-post-merge}, 2022.
\newblock [Accessed 22-09-2024].

\bibitem{berg2022empirical}
Jan~Arvid Berg, Robin Fritsch, Lioba Heimbach, and Roger Wattenhofer.
\newblock An empirical study of market inefficiencies in uniswap and sushiswap.
\newblock In {\em International Conference on Financial Cryptography and Data Security}, pages 238--249. Springer, 2022.

\bibitem{thecryptocortexUnderstandingMarket}
The~Crypto Cortex.
\newblock {U}nderstanding {M}arket {I}nefficiencies and {A}rbitrage in {C}rypto - {T}he {C}rypto {C}ortex --- thecryptocortex.com.
\newblock \url{https://thecryptocortex.com/market-inefficiencies-and-arbitrage/}, 2024.
\newblock [Accessed 22-09-2024].

\bibitem{chi2024remeasuring}
Tianyang Chi, Ningyu He, Xiaohui Hu, and Haoyu Wang.
\newblock Remeasuring the arbitrage and sandwich attacks of maximal extractable value in ethereum.
\newblock {\em arXiv preprint arXiv:2405.17944}, 2024.

\bibitem{wang2022impact}
Ye~Wang, Patrick Zuest, Yaxing Yao, Zhicong Lu, and Roger Wattenhofer.
\newblock Impact and user perception of sandwich attacks in the defi ecosystem.
\newblock In {\em Proceedings of the 2022 CHI Conference on Human Factors in Computing Systems}, pages 1--15, 2022.

\bibitem{flashbotsOrderFlow}
Quintus Kilbourn.
\newblock order flow, auctions and centralisation {I} - a warning | {F}lashbots {W}ritings --- writings.flashbots.net.
\newblock \url{https://writings.flashbots.net/order-flow-auctions-and-centralisation}.
\newblock [Accessed 02-10-2024].

\bibitem{structuralAdvantages}
Max~Resnick Mallesh~Pai.
\newblock {S}tructural {A}dvantages for {I}ntegrated {B}uilders in {M}{E}{V}-{B}oost --- arxiv.org.
\newblock \url{https://arxiv.org/abs/2311.09083}, 2023.

\bibitem{match}
bert.
\newblock {M}{E}{V}-{S}hare: programmably private orderflow to share {M}{E}{V} with users --- collective.flashbots.net.
\newblock \url{https://collective.flashbots.net/t/mev-share-programmably-private-orderflow-to-share-mev-with-users/1264}, 2023.
\newblock [Accessed 23-07-2024].

\bibitem{share}
shea.
\newblock {A}nnouncing {M}{E}{V}-share beta --- collective.flashbots.net.
\newblock \url{https://collective.flashbots.net/t/announcing-mev-share-beta/1650}.
\newblock [Accessed 25-07-2024].

\bibitem{flashbotsMEVShareProgrammably}
Flashbots.
\newblock Mev-share: programmably private orderflow to share mev with users --- collective.flashbots.net.
\newblock \url{https://collective.flashbots.net/t/mev-share-programmably-private-orderflow-to-share-mev-with-users/1264}, 2023.
\newblock [Accessed 10-10-2024].

\bibitem{flashbotsFRP30Quantifying}
{Flashbots Research Proposals}.
\newblock Frp-30: Quantifying shareable mev collective.flashbots.net.
\newblock \url{https://collective.flashbots.net/t/frp-30-quantifying-shareable-mev/1618}, 2023.
\newblock [Accessed 10-10-2024].

\bibitem{deng1994complexity}
Xiaotie Deng and Christos~H Papadimitriou.
\newblock On the complexity of cooperative solution concepts.
\newblock {\em Mathematics of operations research}, 19(2):257--266, 1994.

\bibitem{van2023efficiently}
Tom~C van~der Zanden, Hans~L Bodlaender, and Herbert~JM Hamers.
\newblock Efficiently computing the shapley value of connectivity games in low-treewidth graphs.
\newblock {\em Operational Research}, 23(1):6, 2023.

\bibitem{michalak2013efficient}
Tomasz~P Michalak, Karthik~V Aadithya, Piotr~L Szczepanski, Balaraman Ravindran, and Nicholas~R Jennings.
\newblock Efficient computation of the shapley value for game-theoretic network centrality.
\newblock {\em Journal of Artificial Intelligence Research}, 46:607--650, 2013.

\bibitem{dunePOF}
@dataalways/ Private Order Flow~Monitor.
\newblock Private order flow.
\newblock \url{https://dune.com/dataalways/private-order-flow}, 2024.
\newblock [Accessed 10-12-2024].

\bibitem{10634354}
Fei Wu, Thomas Thiery, Stefanos Leonardos, and Carmine Ventre.
\newblock Strategic bidding wars in on-chain auctions.
\newblock In {\em 2024 IEEE International Conference on Blockchain and Cryptocurrency (ICBC)}, pages 503--511, 2024.

\bibitem{10271857}
Congying Jin, Taotao Wang, Zhe Wang, Long Shi, and Shengli Zhang.
\newblock First-price sealed-bid auction for ethereum gas auction under flashbots.
\newblock In {\em 2023 IEEE International Conference on Metaverse Computing, Networking and Applications (MetaCom)}, pages 449--457, 2023.

\bibitem{chionas2023gets}
Georgios Chionas, Pedro Braga, Stefanos Leonardos, Carmine Ventre, Georgios Piliouras, and Piotr Krysta.
\newblock Who gets the maximal extractable value? a dynamic sharing blockchain mechanism.
\newblock In {\em Proceedings of the 23rd International Conference on Autonomous Agents and Multi-Agent Systems (AAMAS 2024)}, 2024.

\bibitem{mazorra2023towards}
Bruno Mazorra and Nicol{\'a}s Della~Penna.
\newblock Towards optimal prior-free permissionless rebate mechanisms, with applications to automated market makers \& combinatorial orderflow auctions.
\newblock {\em arXiv preprint arXiv:2306.17024}, 2023.

\bibitem{chitra2022improving}
Tarun Chitra and Kshitij Kulkarni.
\newblock Improving proof of stake economic security via mev redistribution.
\newblock In {\em Proceedings of the 2022 ACM CCS Workshop on Decentralized Finance and Security}, pages 1--7, 2022.

\bibitem{roughgarden}
Tim Roughgarden.
\newblock Transaction fee mechanism design.
\newblock {\em Journal of the ACM}, 71(4):1--25, 2024.

\bibitem{damle2024designing}
Sankarshan Damle, Manisha Padala, and Sujit Gujar.
\newblock Designing redistribution mechanisms for reducing transaction fees in blockchains.
\newblock {\em arXiv preprint arXiv:2401.13262}, 2024.

\bibitem{damle2024no}
Sankarshan Damle, Varul Srivastava, and Sujit Gujar.
\newblock No transaction fees? no problem! achieving fairness in transaction fee mechanism design.
\newblock {\em arXiv preprint arXiv:2402.04634}, 2024.

\bibitem{jain2021we}
Anurag Jain, Shoeb Siddiqui, and Sujit Gujar.
\newblock We might walk together, but i run faster: Network fairness and scalability in blockchains.
\newblock {\em arXiv preprint arXiv:2102.04326}, 2021.

\bibitem{chen2024game}
Lin Chen, Lei Xu, Zhimin Gao, Ahmed~Imtiaz Sunny, Keshav Kasichainula, and Weidong Shi.
\newblock A game theoretical analysis of non-linear blockchain system.
\newblock {\em Distributed Ledger Technologies: Research and Practice}, 3(1):1--24, 2024.

\bibitem{siddiqui2020bitcoinf}
Shoeb Siddiqui, Ganesh Vanahalli, and Sujit Gujar.
\newblock Bitcoinf: Achieving fairness for bitcoin in transaction-fee-only model.
\newblock {\em arXiv preprint arXiv:2003.00801}, 2020.

\bibitem{jain2022tiramisu}
Anurag Jain, Sanidhay Arora, Sankarshan Damle, and Sujit Gujar.
\newblock Tiramisu: Layering consensus protocols for scalable and secure blockchains.
\newblock In {\em 2022 IEEE International Conference on Blockchain and Cryptocurrency (ICBC)}, pages 1--3. IEEE, 2022.

\bibitem{damle2021fasten}
Sankarshan Damle, Sujit Gujar, and Moin~Hussain Moti.
\newblock Fasten: Fair and secure distributed voting using smart contracts.
\newblock In {\em 2021 IEEE International Conference on Blockchain and Cryptocurrency (ICBC)}, pages 1--3. IEEE, 2021.

\bibitem{faltings2021orthos}
Boi Faltings and Sujit Gujar.
\newblock Orthos: A trustworthy ai framework for data acquisition.
\newblock In {\em Engineering Multi-Agent Systems: 8th International Workshop, EMAS 2020, Auckland, New Zealand, May 8--9, 2020, Revised Selected Papers}, volume 12589, page 100. Springer Nature, 2021.

\bibitem{srivastava2024decent}
Varul Srivastava and Sujit Gujar.
\newblock Decent-brm: Decentralization through block reward mechanisms.
\newblock {\em arXiv preprint arXiv:2401.08988}, 2024.

\bibitem{narahari2014game}
Yadati Narahari.
\newblock {\em Game theory and mechanism design}, volume~4.
\newblock World Scientific, 2014.

\bibitem{blumrosen2007algorithmic}
Liad Blumrosen and Noam Nisan.
\newblock {\em Algorithmic game theory}.
\newblock Cambridge Univ. Press New York, NY, USA, 2007.

\bibitem{shapley1953value}
Lloyd~S Shapley.
\newblock A value for n-person games.
\newblock {\em Contribution to the Theory of Games}, 2, 1953.

\bibitem{eigenphiDataEigenPhiWisdom}
EigenPhi.
\newblock Mev data | eigen phi --- eigenphi.io.
\newblock \url{https://eigenphi.io/}, 2022.
\newblock [Accessed 1-10-2024].

\end{thebibliography}
% \fi
%%%%%%%%%%%%%%%%%%%%%%%%%%%%%%%%%%%%%%%%%%%%%%%%%%%%%%%%%%%%%%%%%%%%%%%%
\appendix
\section{Matchmaking}
\subsection{Matchmaking vs OFA}
A regular OFA only ensures that private transactions are sold off to the highest bidder but doesn't ensure the timely execution of the transaction. Consider the scenario where the searcher chooses not to send the transactions to a block builder and retain them for the future. While an OFA can introduce a penalty for the searcher that misses the slot, where an OFA waits for a deterministic amount of time and forwards the transaction by itself or auctions it again,  it is hard to identify if the transactions did not make it to execution solely due to searcher and not due to network latency, high competition for blockspace. Hence, potentially enforcing cooperation is not easy. 
Further, the auctioneer can increase the revenue by selling multiple transactions simultaneously as a bundle to searchers rather than individually~\cite{match}. Matchmaking is a recent introduction in the MEV world, where a \emph{matchmaker} aggregates transactions from transaction creators, exposes transaction data to searchers, collects bids along with partial bundles, creates final bundles, and bids to builders for inclusion~\cite{share}. Searchers bid different subsets of transactions, and the matchmaker computes to find optimal allocation. This requires optimizing for the best bundles in a finite amount of time and redistributing the revenue back to the transaction creators to compensate for the value they create.

%%%%%%%%%%%%%%%%%%%%%%%
\subsection*{Desirable Properties of Matchmaking}
% \label{ssec:prop matchmakinng}
%%%%%%%%%%%%%%%%%%%%%%%
We now discuss the desirable properties of matchmaking. Matchmaking involves allocating transactions to searchers and compensating transaction creators. 
Mathematically, let:
   \begin{itemize}
       \item \( v_i \) be the true valuation of searcher \( i \) for the allocation $A \in \mathcal{A}$ (set of possible allocations).
       \item \( b_i \) be the bid of searcher \( i \).
       \item \( A(b) \) be the allocation rule based on all bids \( b = (b_1, b_2, \ldots, b_n) \).
       \item \( p_i(b) \) be the payment rule for searcher \( i \).
       \item \(r_i(b)\) be the reward for transaction creator $i$
   \end{itemize}
  
Some of the generally desired properties of auction mechanisms are:

\textbf{1. Incentive Compatibility (Truthfulness)}:  
A mechanism is incentive-compatible (or truthful) if each searcher's best strategy is to bid their true valuation of the good, regardless of what others are doing. In such a mechanism, searchers have no incentive to misrepresent their preferences.
   The mechanism is incentive-compatible for searchers if:
    \begin{equation}
   v_i(A(b_i, b_{-i})) - p_i(b_i, b_{-i}) \geq v_i(A(b'_i, b_{-i})) - p_i(b'_i, b_{-i}),     
    \end{equation}
   
   for all \( b'_i \), where \( b_{-i} \) are the bids of all other agents.

\textbf{2. Allocative Efficiency (Social Welfare)}  
Allocative Efficiency or Social Welfare Maximizing allocation refers to the allocation that maximizes total value for all searchers, i.e.,
   \begin{equation}
   \max_{A \in \mathcal{A}}\sum_{i=1}^{m} v_i(A(b)).
   \end{equation}

\textbf{3. Individual Rationality}  
A matchmaking $\mathcal{M}$ is individually rational if it is individually rational for both searchers and transaction creators. Each searcher is at least as well off by participating in the matchmaking as they would be if they chose not to participate. Similarly, the utility of each transaction creator is at least as well by including it in matchmaking as they would if they choose not to be involved in matchmaking. Formally, the utility \( u^{\mathcal{S}}_i \) of searcher \( i \) and utility \(u^{\mathcal{T}}_i \) should be non-negative:
   
   \begin{align}
   u^{\mathcal{S}}_i &= v_i(A(b)) - p_i(b) \geq 0.\\
   u^{\mathcal{T}}_i &= r_i \geq 0
   \end{align}
   
\textbf{4. No-deficit}  A matchmaking $\mathcal{M}$ is no-deficit if the total payments collected from the searcher $\mathcal{S}$ equal the total rebates paid to transaction creators of the winning transaction. That is, no money is lost from the process of matchmaking. The condition for no deficit is:
   \begin{equation}
   \sum_{i \in \mathcal{S^{*}}} p_i(b) = \sum_{i \in \mathcal{T^{*}}} r_i.    
   \end{equation}

% In some cases, weaker notions of budget balance (e.g., *weak* or *ex-post* budget balance) are considered, where the auctioneer's deficit or surplus may be constrained rather than exactly zero.

\textbf{5. Fair redistribution} A matchmaking $\mathcal{M}$ is fair to transaction creators if rebate to the transaction creator of each winning transaction $\mathcal{T^{*}}$ is such that (i) $\sum_{i\in\mathcal{T^{*}}} r_i = R$, where $R$ is the revenue generated from searchers, (ii) $\forall i,j$ similar MEV transactions generating same value, $r_i = r_j$, and (iii) $r_i = 0$ for all transactions $i$ that do not create any MEV.

% \section{temp}

% Let $k = d \sqrt{n} + r\ |\  0 \leq r < \sqrt{n}$
% $$
%     \mathcal{B}_k = \{t_i\}, \forall i \in \Gamma_{d,r}
% $$

% $$
%     \Gamma_{d,r} = \left\{ j + \left((r+dj) \textrm{ mod } \sqrt{n} \right) \sqrt{n} \right\}, \forall j \in \{0, \ldots, \sqrt{n}-1\}
% $$

%  OR
% $$
%     \Gamma_{d,r} = \left\{(r\sqrt{n}+dj\sqrt{n} + j) \textrm{ mod } n \right\}, \forall j \in \{0, \ldots, \sqrt{n}-1\}
% $$

% $$
%     \nu(\mathcal{B}_{k_1}) > \nu(\mathcal{B}_{k_2}) \iff r_1 < r_2 \lor (r_1 = r_2 \land d_1 < d_2)  
% $$

\subsection{Winner Determination with Additive Searchers}
\begin{algorithm}
\caption{SPA with VCG Payments}
\begin{algorithmic}[1]
\STATE \textbf{Input:} Bids \( b_{s_1}, b_{s_2}, \dots, b_{s_m} \) 
\STATE Initialize $p_{s_i} = 0$ for each searcher $s_i$ 
\FOR {each transaction $t_j \in T$}
    \STATE Find \( b_{s_i}[t_j] \) of each searcher $s_i$.
    \STATE Let \( {s_i}^* = \arg \max_{s_i} b_{s_i}[t_j] \) \COMMENT{Identify the participant with the highest bid}
    \STATE Let \( b_{s_\text{second}}[t_j] = \max_{s_i \neq {s_i}^*} b_{s_i}[t_j] \) \COMMENT{Find the second-highest bid}
    
    \STATE \textbf{Payments:}
    \FOR {each searcher \( s_i \) }
        \IF {$s_i = {s_i}^*$} 
            \STATE \( p_{s_i} = p_{s_i} + b_{s_\text{second}}[t_j] \) \COMMENT{Winner pays the second-highest bid}
        \ELSE
            \STATE \( p_{s_i}  = p_{s_i} + 0 \) \COMMENT{Non-winners pay nothing}
        \ENDIF
    \ENDFOR
\ENDFOR

\STATE \textbf{Output:} Winner \( {s_i}^* \) and payments \( p_{s_1}, p_{s_2}, \dots, p_{s_n} \)
\end{algorithmic}
\label{algo::spa}
\end{algorithm}

\section{Proving Uniqueness of Marginal Contributions}
\begin{theorem}
    $\forall (a_0, \ldots, a_{n-1}) \in \{-1,0,1\}^n$, the summation $s = \sum_{i=0}^{i=n-1} 3^i*a_i$ is unique.
\end{theorem}
\begin{proof}
    Proof by induction:\\
    Base case: $n=1$
    
    For $a_0 \in \{-1,0,1\}$, the summation $s$ can take value either $-1, 0$ or $1$. Thus, all three possible summations for $n=1$ are unique.\\
    Inductive step: Assume the theorem to be true for $ n = k$.

    The minimum summation possible for $(a_0, \ldots, a_{k-1})$ is when all $a_i = -1$. $$s_{\textrm{min}} = -1 \times \frac{3^k - 1}{3-1} = -\frac{3^k - 1}{2}$$.
    The maximum summation possible for $(a_0, \ldots, a_{k-1})$ is when all $a_i = 1$.
    $$s_{\textrm{max}} = 1 \times \frac{3^k - 1}{3-1} = \frac{3^k - 1}{2}$$

    For $n=k+1$, $a_{k+1}$ can take values either $-1, 0$ or $1$.
    \begin{itemize}
        \item $a_{k+1} = -1$: $s_{\textrm{min}}^0 = -\frac{3^{k+1} - 1}{2}$ and $s_{\textrm{max}}^0 = -\frac{3^k + 1}{2}$.
        \item $a_{k+1} = 0$: $s_{\textrm{min}}^1 = -\frac{3^k - 1}{2}$ and $s_{\textrm{max}}^1 = \frac{3^k - 1}{2}$
        \item $a_{k+1} = 1$: $s_{\textrm{min}}^2 = \frac{3^k + 1}{2}$ and $s_{\textrm{max}}^2 = \frac{3^k - 1}{2}$
    \end{itemize}
    As $s_{\textrm{max}}^0 < s_{\textrm{min}}^1$, $s_{\textrm{max}}^0 < s_{\textrm{min}}^2$ and  $s_{\textrm{max}}^1 < s_{\textrm{min}}^2$, neither ranges of values overlap. Thus, if all possible summations for $n=k$ are unique, then the summations for $n=k+1$ are also unique.\\
    By the principle of mathematical induction, the summations $s = \sum_{i=0}^{i=n-1} 3^i*a_i$ are unique.
\end{proof}
\section{Empirical Analysis}
Figure~\ref{fig:mConts} demonstrates the relationship between the number of unique values and the number of transactions, \(n\). The observed growth rate of the number of unique values closely follows the theoretical prediction of \(\Omega(2^{\sqrt{n}})\). This validates our theory that the number of unique values scales asymptotically as \(\Omega(2^{\sqrt{n}})\) with respect to \(n\), confirming the expected behavior. 

\begin{figure}[H]
\centering
   \begin{subfigure}{0.8\linewidth}
        \includegraphics[height=7cm,width=\textwidth]{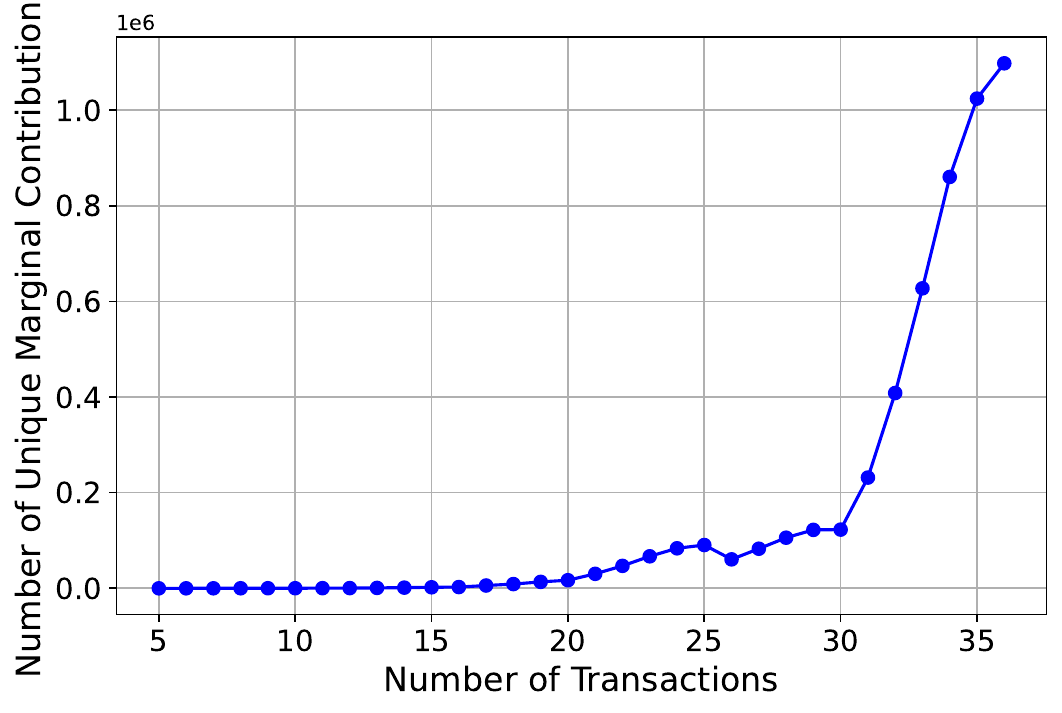}
        \caption{Linear Scale}
    \end{subfigure}

    \begin{subfigure}{0.8\linewidth}                   
        \includegraphics[height=7cm,width=\textwidth]{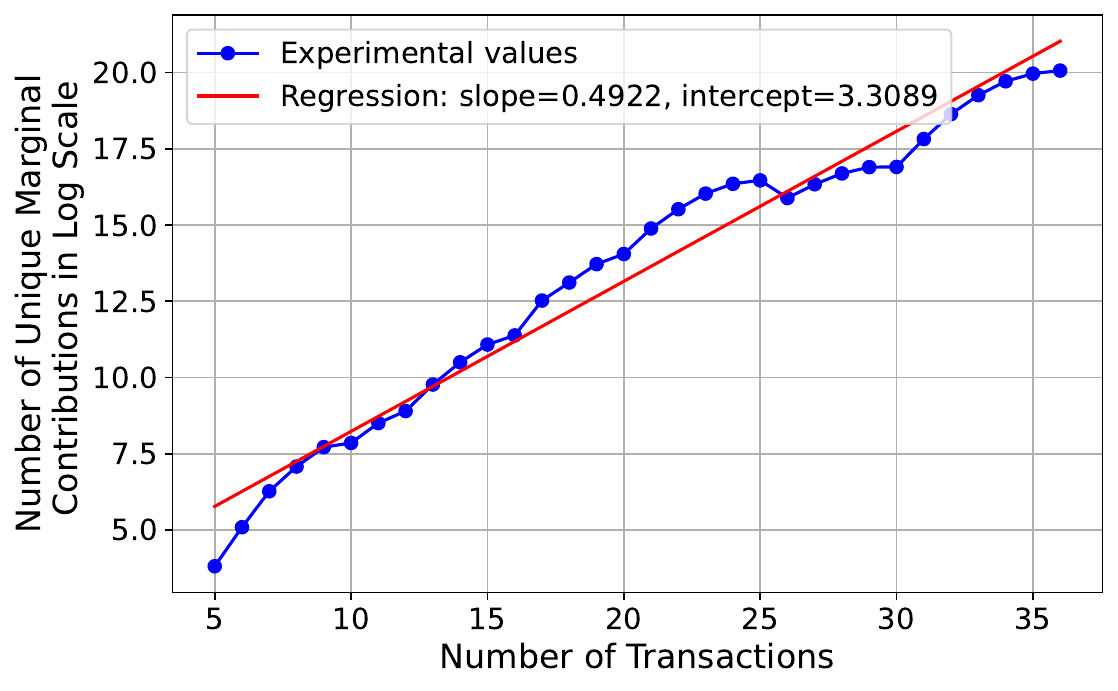}
        \caption{Log Scale}
    \end{subfigure}
    
    \caption{Variation of Unique Marginal Contributions}
    \label{fig:mConts}
\end{figure}
\section{Unianimity Games}
A unanimity game is a fundamental concept in cooperative game theory, characterized by a coalition structure in which a specific subset of players, called a winning coalition, must act unanimously for any payoff to be generated. The game's value depends on whether a coalition contains the required subset.

Let \( N \) be a finite set of players, and let \( S \subseteq N \) be a coalition (a subset of players). A unanimity game is represented by a characteristic function \( v: 2^N \to \mathbb{R} \) defined as:

\[
v(T) =
\begin{cases} 
1, & \text{if } S \subseteq T, \\
0, & \text{otherwise}.
\end{cases}
\]

where \( S \) is the winning coalition, i.e., the minimal subset of players required to generate a value. \( T \) is any coalition under consideration. \( v(T) \) gives the value of the coalition \( T \). It is \( 1 \) if \( T \) contains all players in \( S \), and \( 0 \) otherwise.

The Shapley value of \mevgame\ can derived from unanimity games in the following way. Let $U_S = (\mathcal{T},\omega_S)$ be unanimity game defined on set of transactions $\mathcal{T}$ such that
\begin{equation}\forall T \subseteq \mathcal{T}, \omega_S(T) = 
\begin{cases}
    1\quad S\subseteq T\\
    0\quad \text{Otherwise}
\end{cases}
\end{equation}

\begin{theorem}\
    \label{thm:hars-div}
    $\forall T\subseteq \mathcal{T}$, the value function can be uniquely expressed in terms of unanimity functions in $B_N$ with Harsanyi dividends as the coefficients
    \begin{equation}
        \nu(T) =\sum_{S\in 2^\mathcal{T} \setminus \phi} \Delta_{\nu,S} \omega_S(T)
    \end{equation}
\end{theorem}
\begin{proof}

    Let $E = 2^\mathcal{T}\setminus \phi$ i.e., $E = \{T_1, \ldots, T_{2^{\lvert\mathcal{T}\rvert}-1}\}$
    
    Let $\Lambda$ be the matrix formed using unanimity functions $\Lambda_{(i,j)} = \omega_{T_i}(T_j) \forall i,j\in\{0,\ldots,2^{\lvert\mathcal{T}\rvert}-1\}$ 
    
    \[
       \begin{array}{c|*{7}{c}|}
         \multicolumn{1}{c}{\omega} & \multicolumn{1}{c}{\mathbf{T_1}} & \multicolumn{1}{c}{\mathbf{T_2}} & \multicolumn{1}{c}{\mathbf{T_3}} & \multicolumn{1}{c}{\mathbf{T_4}} & \multicolumn{1}{c}{\cdots} & \multicolumn{1}{c}{\cdots} & \multicolumn{1}{c}{\mathbf{T_{2^{\lvert\mathcal{T}\rvert}-1}}}\\  
        \omega_{T_1} & 1 & 0 & 0 & 1  & \cdots & \cdots &  1 \\ 
        \omega_{T_2} & 0 & 1 & 0 & 1  & \cdots & \cdots &  1 \\ 
        \omega_{T_3} & 0 & 0 & 1 & 0  &\cdots & \cdots  &  1 \\ 
        \omega_{T_4} & 0 & 0 & 0 & 1  & \cdots & \cdots &  1 \\ 
        \omega_{T_5} & 0 & 0 & 0 & 0  & \cdots & \cdots &  1 \\ 
        \omega_{T_6} & 0 & 0 & 0 & 0  & 1 & \cdots &  1 \\ 
        \omega_{T_7} & 0 & 0 & 0 & 0  & 0 & \cdots &  1 \\ 
        \vdots & \vdots & \vdots & \vdots  & \vdots &  \vdots & \vdots & \vdots\\
        \vdots& \vdots & \vdots & \vdots & \vdots &  \vdots & \vdots & \vdots\\
        \omega_{T_{2^{\lvert\mathcal{T}\rvert}-1}} & 0 & 0 & 0 & 0 & 0 & 0 & 1\\ 
    \end{array}
    \]
    
    Assume there is a linear combination of unanimity functions that equals the zero function:
    \[
    \sum_{S\in E} b_S \omega_S(T) = 0, \quad \forall T \in E
    \]
    
    For this to hold, \( b_S \) must be zero for all \( S \subseteq \mathcal{T} \). The unanimity function \( \omega_S(T) \) is nonzero (equal to \( 1 \)) if and only if \( S \subseteq T \). This means \( \omega_S(T) \) activates only those terms where \( S \subseteq T \). If \( b_S \neq 0 \) for some \( S \subseteq \mathcal{T} \), then \( \sum_{S \subseteq \mathcal{T}} b_S \omega_S(T) \neq 0 \) for some \( T \). This contradicts the assumption that the sum equals \( 0 \) for all \( T \). Thus, \( b_S = 0 \) for all \( S \). This proves $\{\omega_S: S\in E\}$ are linearly independent.

    Further, consider the following equations:
    \begin{align*}
        \nu(T_1) &= \sum_{S\in E} b_S \omega_S(T_1)\\
        \nu(T_2) &= \sum_{S\in E} b_S \omega_S(T_2)\\
        \nu(T_3) &= \sum_{S\in E} b_S \omega_S(T_3)\\
        \vdots&\\
         \nu(T_{2^{\lvert\mathcal{T}\rvert}-1}) &= \sum_{S \in E} b_S \omega_S(T_{2^{\lvert\mathcal{T}\rvert}-1})\\
    \end{align*}

    From the $\Lambda$, it can be observed that there are $2^{\lvert\mathcal{T}\rvert}-1$ independent equations and $2^{\lvert\mathcal{T}\rvert}-1$ variables, hence there exists unique solution $\{b_T:T\subseteq\mathcal{T}\}$.
    
    Hence, $B_{\mathcal{T}}= \{\omega_T: T\subseteq \mathcal{T}\setminus\phi\}$ forms the basis for the characteristic function $\nu$. These $b_S$ are nothing but Harsanyi dividends represented as $\Delta_{\nu,S}$.

    $$
      \nu(T) =\sum_{S\in 2^\mathcal{T} \setminus \phi} \Delta_{\nu,S} \omega_S(T)
    $$
\end{proof}
\begin{theorem}
\label{thm:shap-uni}
The Shapley value of $t_j$ in the unanimity game $(\mathcal{T}, \omega_T)$ is:
\[
\varphi_{t_j}(\omega_T) =
\begin{cases}
\frac{1}{|T|}, & \text{if } t_j \in T, \\
0, & \text{if } t_j \notin T.
\end{cases}
\]    
\end{theorem}
\begin{proof}
\[
\Delta \omega_T(C)(t_j) =
\begin{cases}
1, & T \setminus \{t_j\} \subseteq C, \\
0, & T \setminus \{t_j\} \not\subseteq C.
\end{cases}
\]

\[
\implies \forall t_j \notin S, \quad \varphi_{t_j}(\omega_T) = 0.
\]

\[
\omega_T(\mathcal{T}) = 1.
\]
\begin{equation}
    \varphi_{t_j}(\omega_T) = \frac{1}{N!} \sum_{C \subseteq \mathcal{N}_i} |C|! (N - |C| - 1)! \big(\omega_T(C \cup \{t_j\}) - \omega_T(C)\big)
\end{equation}

From symmetry, each transaction in $S$ receives an equal share of $\nu(\mathcal{T})$:
\[
\implies \forall i \in S, \quad \varphi_{t_j}(\omega_T) = \frac{1}{|T|}.
\]    
\end{proof}

\begin{theorem}
Any characterisitic function $\nu: 2^{\mathcal{T}} \setminus \{\emptyset\} \to \mathbb{R}_+$, the Shapley value can be expressed in terms of the Harsanyi dividends as:
\[
\varphi_{t_j}(\nu) = \sum_{T \subseteq 2^{\mathcal{T}} \setminus \{\emptyset\}, t_j \in T} \frac{\Delta_{\nu, T}}{|T|}.
\]
\end{theorem}

\begin{proof}

Let $N_S=|C|! \, (\lvert \mathcal{T} \rvert - |C| - 1)! $
\begin{align*}   
\varphi_{t_j}(\nu) &= \frac{1}{\lvert \mathcal{T} \rvert!} \sum_{C \subseteq \mathcal{T} \setminus \{t_j\}} \, N_C \{\nu(C \cup \{t_j\}) - \nu(C)\}.\\
% \text{Similarly,}&\\
% \varphi_{\omega_S, \mathcal{N}}(i)&= \frac{1}{N!} \sum_{C \subseteq \mathcal{N} \setminus \{i\}} |C|! \, (N - |C| - 1)! \, \{\omega_S(C \cup \{i\}) - \omega_S(C)\}.
 & \text{From Theorem } \ref{thm:hars-div}, \ref{thm:shap-uni}, \text{we have } \\
 &\nu(C) = \sum_{T \subseteq 2^{\mathcal{T}} \setminus \{\emptyset\}} \Delta_{\nu, T} \omega_T(C)\ \text{and}\\
 &\varphi_{t_j}(\omega_T) = \frac{1}{|T|} \quad \text{or } 0.\\
\varphi_{t_j}(\nu) &= \frac{1}{\lvert \mathcal{T} \rvert!} \sum_{C \subseteq \mathcal{T} \setminus \{t_j\}} N_C \big\{ \nu(C \cup \{t_j\}) - \nu(C) \big\} \\
&= \frac{1}{\lvert \mathcal{T} \rvert!} \sum_{C \subseteq \mathcal{T} \setminus \{t_j\}} N_C \Bigg( \sum_{T \subseteq 2^\mathcal{T} \setminus \emptyset} \Delta_{\nu,T} \omega_T(C \cup \{t_j\}) - \\
& \hspace{10em} \sum_{T \subseteq 2^\mathcal{T} \setminus \emptyset} \Delta_{\nu,T} \omega_T(C) \Bigg) \\
&= \frac{1}{\lvert \mathcal{T} \rvert!} \sum_{C \subseteq \mathcal{T} \setminus \{t_j\}} \sum_{T \subseteq 2^\mathcal{T} \setminus \emptyset} N_C \Delta_{\nu,T} \big( \omega_T(C \cup \{t_j\}) - \omega_S(C) \big) \\
&= \frac{1}{\lvert \mathcal{T} \rvert!} \sum_{C \subseteq \mathcal{T} \setminus \{t_j\}} \sum_{T \subseteq 2^\mathcal{T} \setminus \emptyset, t_j \in T} N_C \Delta_{\nu,T} \big( \omega_T(C \cup \{t_j\}) - \omega_T(C) \big)\\ 
% \begin{center}
% $\because \Delta_{\omega_S}(i) = \begin{cases} 
% 1, & \text{if } S \setminus \{i\} \subseteq C, \\
% 0, & \text{otherwise.}
% \end{cases}$
% \end{center}
&= \frac{1}{\lvert \mathcal{T} \rvert!} \sum_{T \subseteq 2^{\mathcal{T}} \setminus \emptyset, i \in T} \sum_{T \subseteq \mathcal{T} \setminus \{t_j\}} N_C \Delta_{\nu, T} \big(\omega_T(C \cup \{i\}) - \omega_T(C)\big)\\
&= \frac{1}{\lvert \mathcal{T} \rvert!} \sum_{T \subseteq 2^{\mathcal{T}} \setminus \emptyset, t_j \in T} \Delta_{\nu, T} \sum_{C \subseteq \mathcal{T} \setminus \{t_j\}} N_C \big(\omega_T(C \cup \{i\}) - \omega_T(C)\big)\\
&= \sum_{T \subseteq 2^{\mathcal{T}} \setminus \emptyset, t_j \in T} \Delta_{\nu, T} \Bigg( \frac{1}{\lvert \mathcal{T} \rvert!} \sum_{C \subseteq \mathcal{T}\setminus{t_j}} N_C \big(\omega_T(C \cup \{t_j\}) - \omega_T(C)\big) \Bigg)\\
&= \sum_{T \subseteq 2^{\mathcal{T}} \setminus \emptyset, t_j \in T} \Delta_{\nu, T} \varphi_{t_j}(\omega_T)\\
&= \sum_{T\subseteq 2^\mathcal{T}\setminus \phi, t_j \in T} \frac{\Delta_{\nu,T}}{\lvert T \rvert}
\end{align*}
\end{proof}
%%% Remove comment to use the external .bib file (using bibtex).


\end{document}